%% file: CLTLSynthesis.tex
\documentclass[a4paper,UKenglish,cleveref, autoref, 
thm-restate]{lipics-v2021}

\pdfoutput=1 

\usepackage{xspace,amsmath,graphicx}
\usepackage{tikz}

\hideLIPIcs  

\title{Realizability Problem for Constraint LTL} 

\author{Ashwin Bhaskar}{Chennai Mathematical Institute, India 
}{}{}{}

\author{M.~Praveen}{Chennai Mathematical Institute, 
India \and CNRS IRL ReLaX, India}{}{}{This author is partially supported 
by the Infosys foundation}

\authorrunning{Ashwin Bhaskar and M.~Praveen} 

\Copyright{Ashwin Bhaskar and M.~Praveen} 



\category{} 

\relatedversion{} 

\nolinenumbers 

\EventEditors{John Q. Open and Joan R. Access}
\EventNoEds{2}
\EventLongTitle{42nd Conference on Very Important Topics (CVIT 2016)}
\EventShortTitle{CVIT 2016}
\EventAcronym{CVIT}
\EventYear{2016}
\EventDate{December 24--27, 2016}
\EventLocation{Little Whinging, United Kingdom}
\EventLogo{}
\SeriesVolume{42}
\ArticleNo{23}
\tikzstyle{state}=[draw=black, fill=black, circle, inner 
sep=0.01cm, minimum
height=0.1cm]
\usetikzlibrary{decorations.pathreplacing,fit}

\newcommand{\set}[1]{\{#1\}}
\newcommand{\Nat}{\mathbb{N}}
\newcommand{\Int}{\mathbb{Z}}

\newcommand{\nxt}{X}
\newcommand{\until}{U}
\newcommand{\gap}{\mathit{gp}}
\newcommand{\gf}{G}

\newcommand{\lav}{V^a}
\newcommand{\lat}[1]{T^a[#1]}

\newcommand{\fbv}{V^b}
\newcommand{\slav}{\mathit{SV}^a}

\newcommand{\sfbv}{\mathit{SV}^b}
\newcommand{\efbv}{\mathit{EV}^b}
\newcommand{\sv}{\mathit{SV}}
\newcommand{\ev}{\mathit{EV}}

\newcommand{\smap}{\mathit{sm}}
\newcommand{\emap}{\mathit{em}}
\newcommand{\maplabel}{\mathit{L}}
\newcommand{\Map}{M}
\newcommand{\eMap}{\mathit{EM}}
\newcommand{\sMap}{\mathit{SM}}
\newcommand{\sstrat}{\mathit{st}}
\newcommand{\estrat}{\mathit{et}}
\newcommand{\ceil}[2]{\left\lceil #1 \right\rceil_{#2}}
\newcommand{\sm}{\mathrm{\mu}}
\newcommand{\csys}{\mathcal{D}}
\newcommand{\pref}{\upharpoonright}
\newcommand{\frames}{\mathcal{F}}
\newcommand{\frs}{\mathrm{frames}}
\newcommand{\maps}{\mathrm{maps}}
\newcommand{\pth}{\pi}
\newcommand{\node}{\eta}
\newcommand{\tree}{T}
\newcommand{\restr}{\upharpoonright}
\newcommand{\restslav}{\restr \slav}
\newcommand{\restsv}{\restr \sv}
\newcommand{\restsfbv}{\restr \sfbv}
\newcommand{\restefbv}{\restr \efbv}
\newcommand{\restfbv}{\restr \fbv}
\newcommand{\lo}{\sqsubseteq}
\newcommand{\lto}{\sqsubset}
\newcommand{\ro}{\lto^*}
\newcommand{\rto}{\lto^+}
\newcommand{\go}{\sqsupset^*}
\newcommand{\gto}{\sqsupset^+}
\newcommand{\prompt}{\mathbf{F_P}}
\newcommand{\splx}{x^a_{c}}
\newcommand{\sply}{y^a_{c}}
\newcommand{\col}{p}
\newcommand{\pfr}{\mathit{pf}}

\newcommand{\sys}{{\tt system}\xspace}
\newcommand{\Sys}{{\tt System}\xspace}
\newcommand{\env}{{\tt environment}\xspace}

\begin{document}

\maketitle

\begin{abstract}
Constraint linear-time temporal logic (CLTL) is an extension of LTL that 
is interpreted on sequences of valuations of variables over an infinite 
domain. The atomic formulas are interpreted as constraints on the 
valuations. The atomic formulas can constrain valuations over a range of 
positions along a sequence, with the range being bounded by a 
parameter depending on the formula. The satisfiability and model 
checking problems for CLTL have been studied by Demri and D’Souza. 
We consider the 
realizability problem for CLTL. The set of variables is partitioned into 
two parts, with each part controlled by a player. Players take turns to 
choose valuations for their variables, generating a sequence of 
valuations. The winning condition is specified by a CLTL formula---the first player wins if the sequence of valuations satisfies the specified 
formula. We study the decidability of checking whether the first player 
has a winning strategy in the realizability game for a given CLTL 
formula. We prove that it is decidable in the case where the domain 
satisfies the completion property, a property introduced by Balbiani and 
Condotta in the context of satisfiability. We prove that it is undecidable 
over $(\Int,<,=)$, the domain of integers with order and equality. We 
prove that over $(\Int,<,=)$, it is decidable if the atomic constraints in 
the formula can only constrain the current valuations of variables 
belonging to the second player, but there are no such restrictions for 
the variables belonging to the first player. We call this single-sided 
games.

Prompt-LTL is an extension of LTL with the prompt-eventually operator,  which imposes a bound on the wait time for all prompt-eventually sub-formulas. CLTL can be similarly extended to 
prompt-CLTL. We prove that decidability is maintained for single-sided 
games, even if we allow prompt-CLTL formulas.
\end{abstract}

\section{Introduction}
\label{sec:intro}
\input{intro}

\section{Preliminaries}
\label{sec:prelim}
\input{prelim}

\section{Undecidability over Integers with Order and Equality}
\label{sec:undec}
\input{undecidable}

\section{Symbolic Models}
\label{sec:symbModels}
\input{symbModel}

\section{Decidability Over Domains Satisfying the Completion Property}
\label{sec:decidableCompletion}
\input{completion}

\section{Decidability of single-sided CLTL games over $(\Int, <, =)$}
\label{sec:decSingleSided}
\input{singleSided}
\input{treeAutomaton}

\section{Decidability of single-sided prompt-CLTL games over $(\Int, <, =)$}
\label{sec:decpromptCLTL}
\input{promptCLTL}

\section{Discussion and Future Work}
\label{conclusion}
\input{conclusion}

\bibliography{references}
\end{document}

%% file: intro.tex



Propositional linear temporal logic (LTL) and related automata theoretic 
models have been extended in various ways to make it more 
expressive. Prompt-LTL \cite{KPV09}, Constraint LTL \cite{DD07}, LTL 
with freeze operators \cite{DL09}, temporal logic of repeating 
values \cite{DDG12, PFD16}, finite memory automata \cite{KF94}, data 
automata \cite{BDMSS11} are all examples of this. Prompt-LTL is concerned 
with bounding wait times for formulas that are intended to become true 
eventually, while other extensions are concerned with using variables 
that range over infinite domains in place of Boolean propositions 
used in propositional LTL. Variables ranging over infinite domains are a 
natural choice for writing specifications for systems that deal with 
infinite domains. For example, constraint LTL has been used for 
specifications of cloud based elastic systems \cite{BBDGGK14}, where 
the domain of natural numbers are used to reason about the number of 
resources that are being used by cloud based systems.

An orthogonal development in formal verification is synthesis, that is 
concerned with automatically synthesizing programs from logical 
specifications. The problem was identified by Church \cite{C64} and 
one way to solve it is by viewing it as the solution of a two person 
game. For specifications written in propositional LTL, the worst case 
complexity of the realizability problem is doubly exponential 
\cite{PR89}. However, efficient algorithms exist for fragments of LTL. The algorithms are efficient enough and the 
fragments are expressive enough to be used in practice, for example to 
synthesize robot controllers \cite{KFP09}, data buffers 
and data buses \cite{PPS06}.

This paper is in an area that combines both developments 
mentioned in the above paragraphs. We consider constraint LTL (CLTL) 
and
partition the set of variables into two parts, each being owned by a 
player in a two player game. The players take turns to choose a valuation 
for their variables over an infinite domain. The game is played forever 
and results in a sequence of valuations. The first player tries to ensure 
that the resulting sequence satisfies a specified CLTL formula (which is 
the winning condition) and the second player tries to foil this. We study 
the decidability of checking whether the first player has a winning 
strategy, called the realizability problem in the sequel. CLTL is parameterized by a constraint system, that can have 
various relations over the infinite domain. The atomic formulas of CLTL 
can compare values of variables in different positions along a range of 
positions, using the relations present in the constraint system. The 
range of positions is bounded and depends on the formula. E.g., an 
atomic formula can say that the value of $x$ at a position is less than 
the value of $y$ in the next position, in the domain of integers or real 
numbers with linear order. Decidability of the CLTL realizability problem depends on the constraint system. It also depends 
on whether the atomic formulas can compare values at different 
positions of the input, as opposed to comparing values of different 
variables at the same position of the input. If the former is allowed only 
for variables belonging 
to one of the players, they are called single-sided games. This is 
illustrated next.

For instance in cloud based elastic systems \cite{BBDGGK14}, 
the number of resources in use is tracked with respect to the number of 
virtual machines running. One desirable property is that if the number 
of 
virtual machines increases, the number of resources allocated must also 
increase. Typically the number of resources allocated is under the system's 
control and the number of virtual machines is under the environment's 
control. Specifying this property will require comparing the number of 
currently allocated resources with the same number in the next 
position. We may also compare the current number of virtual 
machines with the same number at the next position, but this will need 
the both the system and the environment to compare numbers at different 
positions. Instead, if we only allow the environment to decide whether a 
new virtual machine request is raised at the current position, the game will be 
single-sided.

\textbf{Contributions} We prove that the realizability problem is
\begin{enumerate}
	\item \textsc{2EXPTIME}-complete for CLTL over constraint systems that satisfy a so-called 
	completion property,
	\item undecidable for CLTL over integers with linear order and equality and
	\item \textsc{2EXPTIME}-complete for CLTL single-sided games on integers with linear 
	order and equality.
	\item \textsc{2EXPTIME}-complete for prompt-CLTL single-sided games on integers 
	with linear order and equality.
\end{enumerate}
The third result above is the main one and is inspired by concepts used 
in satisfiability \cite{DD07}. In satisfiability, this technique is based on 
patterns that repeat in ultimately periodic words. It requires new 
insights to make it work in trees that we use to represent strategies 
here. 

\textbf{Related works} Two player games on automata models and  
logics dealing with infinite domains have been studied before \cite{RFE21,EFK20}. The techniques involved are similar to those used here in the sense that instead of reasoning about sequences of values from an infinite domain, sequences of elements from a finite abstraction are considered. Single-sided games are considered in \cite{EFK20}, like we do here, but for register automata specifications. Their result subsumes ours, since register automata are more expressive than CLTL. In register automata, values can be compared even if they occur far apart in the input sequence, but in CLTL, values can only be compared if they occur within a bounded distance. For this reason, CLTL can be handled with simpler arguments, resulting in some differences in technical details, which we will highlight later in this paper. This can potentially speed up procedures in case the specifications only need CLTL and not the full power of register automata\footnote{This does need a detailed study, which we defer to future work.}. 
 Similar single-sided games are also considered in \cite{PMF20}, for an 
extension of LTL incomparable with CLTL. There, single-sided games 
are reduced to energy games \cite{AMSS13} to get decidability.

Even apart from infinite alphabets, the synthesis problem continues to 
be actively under research. In \cite{FMMV21}, 
environment is assumed to satisfy some properties, 
and the system is expected to guarantee that it satisfies some 
properties. Both the assumptions and guarantees can be written in 
prompt-LTL and it is shown that theoretically, the complexity does not 
increase from the case of plain prompt-LTL synthesis (i.e. without 
assumptions and guarantees). We consider prompt-LTL without 
assumptions and 
guarantees, 
but we consider infinite domains.

%% file: prelim.tex
Let $\Int$ be the set of integers and $\Nat$ be the set of non-negative 
integers. We denote by $\ceil{i}{k}$ the number $i$ ceiled at $k$: 
$\ceil{i}{k} = i$ if $i \le k$ and $\ceil{i}{k}=k$ otherwise. If $m$ is any 
mapping and $S$ is a subset of the domain of $m$, we 
denote by $m \restr S$ the mapping $m$ restricted to the domain $S$. 
For a sequence of mappings $m_1 \cdot m_2 \cdots$, we write $m_1 
\cdot m_2 \cdots \restr S$ for $m_1 \restr S \cdot m_2 \restr S 
\cdots$. For integers $n_1,n_2$, we denote by $[n_1,n_2]$ the set 
$\set{n \in \Int \mid n_1 \le n \le n_2}$.

We recall the definitions of constraint systems and constraint LTL 
(CLTL) 
from 
\cite{DD07}. A constraint system $\csys$ is of the form $(D, R_1,\ldots, 
R_n, \mathcal{I})$, where $D$ is a non-empty set called the 
domain. Each $R_i$ is a predicate symbol of arity $a_i$, with 
$\mathcal{I}(R_i) \subseteq D^{a_i}$ being its interpretation.

Let $V$ be a set of variables, partitioned into the sets $\lav,\fbv$ of 
look-\textbf{a}head and future-\textbf{b}lind variables. A look-ahead term is of the 
form $\nxt^i y$, where $y$ is a look-ahead variable, $i \ge 0$ and 
$\nxt$ is a symbol intended to denote ``next''. For $k \ge 0$, we 
denote by $\lat{k}$ the set of all look-ahead terms of the form 
$\nxt^i 
y$, where $i \in [0,k]$ and $y$ is a look-ahead variable. A 
constraint 
$c$ is of the form $R(t_1,\ldots, t_n)$, where $R$ is a predicate 
symbol 
of arity $n$ and $t_1,\ldots, t_n$ are all future-blind variables or they 
are all look-ahead terms. The syntax of CLTL is given by the following 
grammar, where $c$ is a constraint as defined above.
\begin{align*}
	\phi ::= c ~|~ \lnot \phi ~|~ \phi \lor \phi ~|~ \nxt \phi ~|~ \phi 
	\until \phi
\end{align*}

The semantics of CLTL are defined over sequences $\sigma$ (also 
called concrete models in the following); for every $i 
\ge 
0$, $\sigma(i) \colon V \to D$ is a mapping of the variables. Given, $x_1,\ldots, 
x_n \in \lav$ and $i_1,\ldots,i_n \in \Nat$, the 
$i$\textsuperscript{th} position of a concrete model $\sigma$ 
satisfies 
the constraint $R(\nxt^{i_1}x_1, \ldots, \nxt^{i_n}x_n)$ (written as 
$\sigma,i \models R(\nxt^{i_1}x_1, \ldots, \nxt^{i_n}x_n)$) if 
$(\sigma(i+i_1)(x_1), \ldots, \sigma(i+i_n)(x_n)) \in \mathcal{I}(R)$. If 
the constraint is of the form $R(x_1,\ldots, x_n)$ where $x_1,\ldots, 
x_n \in \fbv$, then $\sigma,i \models R(x_1, 
\ldots, x_n)$ if 
$(\sigma(i)(x_1), \ldots, \sigma(i)(x_n)) \in \mathcal{I}(R)$. The 
semantics is extended to the rest of the syntax similar to the usual 
propositional LTL. We use the 
standard abbreviations $F\phi$ (resp.~$G \phi$) to mean that $\phi$ 
is 
true at some position (resp.~all positions) in the future. The 
$\nxt$-length of a look-ahead term $\nxt^i y$ is $i$. We say that a 
formula is 
of $\nxt$-length $k$ if it uses look-ahead terms of $\nxt$-length at 
most $k$. The 
constraint 
system $(\Int,<,=)$ (resp.~$(\Nat,<,=)$) has the domain $\Int$ 
(resp.~$\Nat$) and $<,=$ are interpreted as the usual linear order 
and 
equality relations. The formula $G(x<\nxt y)$ will be true in the first 
position of a concrete model if in all positions, the value of $x$ is less 
than the value of $y$ in the next position.

We adapt the concept of realizability games \cite{PR89} to CLTL. 
There are two players \sys and \env. The set of variables $V$ 
is 
partitioned into two parts $\sv,\ev$ owned by \sys, \env respectively. 
The \env begins by choosing a mapping $\emap_0 \colon \ev \to D$, 
to 
which \sys responds by choosing a mapping $\smap_0 \colon \sv \to D$. 
This 
first round results in the mapping $\emap_0 \oplus 
\smap_0$. This notation is used to define the function such that 
$\emap_0 
\oplus \smap_0(x)=\emap_0(x)$ if $x\in \ev$ and $\emap_0 \oplus 
\smap_0(x)=\smap_0(x)$ if $x\in \sv$. In the next round, the 
two players chose mappings $\emap_1,\smap_1$. Both the players continue to play forever and the play results in a concrete model $\sigma = (\emap_0 
\oplus \smap_0) (\emap_1 \oplus \smap_1) \cdots$. The winning 
condition 
is specified by a CLTL formula $\phi$. \Sys wins this play of the game if 
$\sigma,0 \models \phi$.

Let $\Map$ (resp.~$\eMap$,$\sMap$) be the set of all mappings of 
the 
form $V \to D$ (resp.~$\ev \to D$, $\sv \to D$). For a concrete model 
$\sigma$ and $i\ge 0$, let $\sigma \pref i$ denote the prefix of 
$\sigma$ of length $i$ (for $i=0$, $\sigma \pref i$ is the empty 
sequence $\epsilon$). An \env strategy is a function 
$\estrat \colon \Map^* \to \eMap$ and a \sys strategy is a function 
$\sstrat \colon \Map^* \cdot \eMap \to \sMap$. We say that the \env 
plays according to the strategy $\estrat$ if the resulting model $\sigma 
= (\emap_0 \oplus \smap_0) (\emap_1 \oplus \smap_1) \cdots $ is 
such that $\emap_i = \estrat(\sigma\pref i)$ for all $i \ge 0$. \Sys 
plays according to the strategy $\sstrat$ if the resulting model $\sigma 
= (\emap_0 \oplus \smap_0) (\emap_1 \oplus \smap_1) \cdots $ is 
such that $\smap_i=\sstrat(\sigma\pref i\cdot \emap_i)$ for all $i 
\ge 0$. We say that $\sstrat$ is a winning strategy for \sys if she wins 
all plays of the game played according to $\sstrat$, irrespective of the strategy 
used by \env. For example, let us consider a CLTL game with $V = \lav = \set{x,y}, \ev = \set{x}, \sv = \set{y}$, over the constraint system $(\Int,<,=)$ with winning condition $G ((y > \nxt y) \wedge \neg (({\nxt}^2 x > y) \wedge ({\nxt}^2 x < \nxt y)))$. For \sys to win, the sequence of valuations for $y$ should form a descending chain, and at any position, the value of $x$ should be outside the interval defined by the previous two values of $y$. \Sys has a winning strategy in this game: it can choose $y$ to be $-i$ in the $i^{\textnormal{th}}$ round and the \env cannot choose its $x$ to be strictly between the previous two values of $y$ in any round. \Sys does not have a winning strategy in the same game when it is considered over $(\Nat, <, =)$, as there is no infinite descending sequence of natural numbers. \Sys does not have a winning strategy over dense domains, since \env can choose the third value of $x$ to be strictly between the first two values of $y$, violating the winning condition. Given a CLTL formula $\phi$, the \textbf{realizability problem} is to check whether \sys has a winning 
strategy in the CLTL game whose winning condition is 
$\phi$.

%% file: undecidable.tex
The realizability problem is undecidable for CLTL over 
$(\Int, <, =)$ and $(\Nat, <, =)$. We prove this by a reduction from the 
repeated control state reachability problem for 2-counter machines, which is known to be undecidable \cite{AH94}. 
The undecidability result holds even in the case where only the \env can own future-blind variables (unboundedly many) and both \sys and \env are restricted to own only a single look-ahead variable each.
The main idea of the reduction is that one of the players 
simulates the counter machine and the other player catches mistakes, 
like other similar reductions 
for games \cite{ABd03}. To make this work when each player owns  a 
single 
look-ahead variable, we use the following idea. Let $x$ be the look-ahead variable owned by the \env player, 
and $y$ be the look-ahead variable owned by the \sys player. The first counter is updated only at odd rounds of a play of the game,
and it is set to be the difference between $x$ and $y$ in that round. Similarly, the second counter is updated only at even rounds of a play of the game and is set to be the difference between $x$ and $y$ in the that round. So both players participate in the
simulation. The \env additionally 
chooses the next counter machine transition to be executed at every odd round using its future-blind variables. \Sys has the additional responsibility of catching mistakes. The CLTL 
formulas 
specifying the winning condition ensure that any player 
who doesn't fulfill their responsibility loses.

We now state and prove the result formally. 

\begin{theorem}
	\label{thm:undecidable}
	Checking whether \sys has a winning strategy in CLTL games over $(\Int,<,=)$ or $(\Nat,<,=)$ where 
	both players have only one look-ahead variable each and \env 
	additionally has unboundedly many future-blind variables is undecidable.
\end{theorem}

\begin{proof}
Given a 2-counter machine, we design an instance of the CLTL realizability problem. Let $x$ be the look-ahead variable owned by the \env player, 
and $y$ be the look-ahead variable owned by the \sys player. 
Corresponding to every transition $t$ of the counter machine, let $u_t$ 
and $v_t$ be two future-blind variables owned by the \env 
such that at every position of the concrete model built during the game, 
$u_t = v_t$ would indicate that the transition $t$ is taken at that 
position and $\neg(u_t = v_t)$ indicates that transition $t$ is not taken 
at that position. For the sake of notational convenience, we shall 
assume that corresponding to each such $t$, the \env player 
has a Boolean variable $p_t$ and we enforce that at each position $p_t$ 
holds iff $(u_t = v_t)$ holds.

Now, using the variables $x$ and $y$, we shall encode the counter 
$c_1$ at the even positions and the counter $c_2$ at the odd positions. 
The role of the \env is to choose a transition of the counter 
machine and then, both players together participate in ensuring that the 
counters are updated according to the transition chosen. In addition, 
the \sys player (using the variable $y$) also ensures that the 
\env player neither makes an illegal transition nor updates the 
counters incorrectly.

Suppose $\sigma$ is the concrete model built during a game. We 
denote by $x_i$ the value of $x$ at position $i$. Similarly for $y_i$. 
The 
initial value of the counter $c_1$ is given by $x_0 - y_0$ and the initial 
value of $c_2$ is given by $x_1 - y_1$. The $i^{\textnormal{th}}$ 
transition chosen by the \env is in the $2i^{\textnormal{th}}$ 
position of the valuation sequence of the concrete model. The value of 
the counter $c_1$ just after taking the $i^{\textnormal{th}}$ transition 
is updated at the $2i^{\textnormal{th}}$ position and is given by 
$(x_{2i} - y_{2i})$ and the value of the counter $c_2$ just after taking 
the $i^{\textnormal{th}}$ transition is updated at the 
$(2i+1)^{\textnormal{st}}$ position and is given by $(x_{2i+1} - 
y_{2i+1})$. Let the set $\Phi_e$ consist of the following formulas, each 
of which denotes a 'mistake' made by the \env player.

\begin{itemize}
    \item The \env chooses some transition in either the 
    $0^{\textnormal{th}}$ or the $1^{\textnormal{st}}$ position.
    
    \begin{equation*}
        (\bigvee_{t\textnormal{ is a any transition}}p_t) \vee {\nxt}( 
        \bigvee_{t\textnormal{ is a any transition}}p_t) 
    \end{equation*}
    
    \item First transition chosen, is not an initial transition
    
    \begin{equation*}
        {\nxt}^2 (\bigvee_{t\textnormal{ is not an initial transition}}p_t)
    \end{equation*}
    
    \item \env chooses more than one transition at some position
    \begin{equation*}
        {\nxt}^2F (\bigvee_{t \not = t'} (p_t \wedge p_{t'}))
    \end{equation*}
    
    \item \env chooses a transition at two consecutive positions
    \begin{equation*}
        {\nxt}^2F (\bigvee_{t, t'} (p_t \wedge {\nxt}p_{t'}))
    \end{equation*}
    
    \item \env does not choose any transition at two consecutive 
    positions (other than the $0^{\textnormal{th}}$ and 
    $1^{\textnormal{st}}$ positions)
    \begin{equation*}
        {\nxt}F(\bigwedge_t (\neg p_t) \wedge \bigwedge_t {\nxt}(\neg p_t))
    \end{equation*}
    
    \item Consecutive transitions are not compatible
    
    \begin{equation*}
        {\nxt}^2F (\bigvee_{t'\textnormal{ cannot come after }t} (p_t \wedge 
        {\nxt}^2p_{t'}))
    \end{equation*}
    
    \item A transition increments $c_1$ but the value of $x$ does not 
    increase.
    
    \begin{equation*}
        F(\bigvee_{t\textnormal{ increments }c_1} {\nxt}^2 p_t \wedge \neg (x 
        < {\nxt}^2x))
    \end{equation*}
    
    \item A transition increments $c_2$ but the value of $x$ in the 
    position just after the transition does not increase.
    
    \begin{equation*}
        F(\bigvee_{t\textnormal{ increments }c_2} {\nxt} p_t \wedge \neg (x < 
        {\nxt}^2x))
    \end{equation*}
    
    \item A transition decrements $c_1$ but the value of $x$ does not 
    decrease.
    
    \begin{equation*}
        F(\bigvee_{t\textnormal{ decrements }c_1} {\nxt}^2 p_t \wedge \neg 
        (x > {\nxt}^2x))
    \end{equation*}
    
    \item A transition decrements $c_2$ but the value of $x$ in the 
    position just after the transition does not decrease.
    
    \begin{equation*}
        F(\bigvee_{t\textnormal{ decrements }c_2} {\nxt} p_t \wedge \neg (x > 
        {\nxt}^2x))
    \end{equation*}
    
    \item A transition demands that the counter $c_1$ remain same, but 
    the value of $x$ changes at that position.
    
    \begin{equation*}
        F(\bigvee_{t\textnormal{ does not change }c_1} {\nxt}^2 p_t \wedge 
        \neg (x = {\nxt}^2x))
    \end{equation*}
    
    \item A transition demands that the counter $c_2$ remain same, but 
    the value of $x$ changes at that position.
    
    \begin{equation*}
        F(\bigvee_{t\textnormal{ does not change }c_2} {\nxt} p_t \wedge \neg 
        (x = {\nxt}^2x))
    \end{equation*}
    
    \item A transition tests that the value of $c_1$ is zero but the value 
    of $x$ at that position does not equal the value of $y$
    
    \begin{equation*}
        F(\bigvee_{t\textnormal{ tests }c_1 = 0} p_t \wedge \neg (x = y))
    \end{equation*}
    
    \item A transition tests that the value of $c_2$ is zero but in the next 
    position, the value of $x$ does not equal the value of $y$
    
    \begin{equation*}
        F(\bigvee_{t\textnormal{ tests }c_2 = 0} p_t \wedge \neg {\nxt}(x = y))
    \end{equation*}
    
    \item The value of a counter is negative at some position.
    
    \begin{equation*}
        F(x < y)
    \end{equation*}
    
    \item Either $c_1$ or $c_2$ at some position, is incremented or 
    decremented by more than 1.
    
    \begin{equation*}
        F (((x < {\nxt}^2 y) \wedge ({\nxt}^2 y < {\nxt}^2 x)) \vee (({\nxt}^2 x < {\nxt}^2 y) 
        \wedge ({\nxt}^2 y < x)))
    \end{equation*}
    
    \item For some transition $t$, at some position, $p_t$ is true but 
    $u_t \ne v_t$ or vice versa.
    
    \begin{equation*}
        F (\bigvee_{t}\neg ((u_t = v_t) \iff p_t))
    \end{equation*}
\end{itemize}

The set $\Phi_s$ consists of the following formulas, each of which 
denotes constraints that the \sys has to satisfy.

\begin{itemize}
    \item The value of $y$ is same in the first three positions.
    
    \begin{equation*}
        (y = {\nxt}y) \wedge {\nxt} (y = {\nxt}y)
    \end{equation*}
    
    \item The initial value of the counter $c_1$ is positive. (This 
    combined with the previous constraint automatically ensures that the 
    initial value of counter $c_2$ is also positive)
    
    \begin{equation*}
        y \geq x
    \end{equation*}
    
    \item The \sys ensures that either the value of $y$ remains 
    constant throughout or it is used at some position to catch a mistake 
    (of incrementing or decrementing a counter value) made by the 
    \env
    
    \begin{equation*}
        G(y = {\nxt}y) \vee F(((x < {\nxt}^2 y) \wedge ({\nxt}^2 y < {\nxt}^2 x)) \vee (({\nxt}^2 
        x < {\nxt}^2 y) \wedge ({\nxt}^2 y < x)))
    \end{equation*}
    
    \item If the \env has not made a mistake at any point, then, 
    the halting state is reached infinitely often.
    
    \begin{equation*}
        G (y = {\nxt}y) \implies GF (\bigvee_{t\textnormal{ is halting}} p_t)
    \end{equation*}
\end{itemize}
The winning condition of the CLTL game is given by:  

\begin{equation*}
	\bigvee \Phi_e \vee \bigwedge \Phi_s
\end{equation*}
    
For the \sys player to win, either one of the formulas in $\Phi_e$ 
	must hold or all the formulas in $\Phi_s$ must be true. Hence, for 
	the 
	\sys to win, the \env should make a mistake during the 
	simulation or both players together simulate the 2-counter machine 
	correctly and the halting state is reached infinitely often. Thus, the 
	\sys has a winning strategy if and only if the 2-counter machine 
	reaches the halting state.
\end{proof}

%% file: symbModel.tex
The models of CLTL are infinite sequences over infinite alphabets. Frames, 
introduced in 
\cite{DD07}, abstract them to finite alphabets. We adapt frames to 
constraint systems of the form $(D,<,=)$. Conceptually, frames and symbolic models as we will define here are almost the same as introduced in \cite{DD07}, where the authors used these notions to solve the satisfiability problem for CLTL. For the purpose of CLTL games, we use slightly different definitions and notations, as this makes it easier to present game-theoretic arguments. For the rest of the paper, we shall assume that the set of variables $V$ is finite. And unless mentioned otherwise, we shall assume that $D$ is $\Int$, $\Nat$ or a domain that satisfies a so-called completion property.

Suppose that the first player owns the variables $x,z$. The second player 
owns $y$ and wants to ensure that $x < y ~\land ~ y 
< z$ over the domain of integers. It depends on whether 
the gap between the values assigned by the first player to $x$ and to $z$, is large enough for the 
second player 
to push $y$ in between.
\begin{definition}[gap functions]
	\label{def:gapFunctions}
	Given a mapping $m \colon \fbv \to D$, we associate with it a gap function 
	$\gap \colon \fbv \to \Nat$ as follows. Arrange $\fbv$ as $x_0, x_1, 
	\ldots$ 
	such that $m(x_0) \le m(x_1)\le \cdots$. Define the function $\gap$ 
	such that $\gap(x_0)=0$ and 
	$\gap(x_{l+1}) = \gap(x_l) + \ceil{m(x_{l+1}) - 
		m(x_l)}{|\fbv|-1}$ for all $l <|\fbv|-1$. 
\end{definition}
The left hand side of 
the above equation denotes 
the gap between $x_l$ and $x_{l+1}$ according to the $\gap$ function. 
The right hand side denotes the gap 
between the same variables according to the mapping $m$, ceiled at 
$|\fbv|-1$. Since, $\fbv$ is finite, the set of gap functions is also finite.
We use gap functions only for future-blind variables $\fbv$, only for the domains $\Int$ or $\Nat$. Hence, the minus sign '$-$' in the definition of gap functions is interpreted as the usual subtraction over $\Int$ or $\Nat$.

The following 
definition formalizes how a frame captures information about orders 
and gaps for $s$ 
successive positions.
\begin{definition}[Frames]
	\label{def:frames}
	Given a number $s \ge 1$, an $s$-frame $f$ is a pair 
	$(\le_f,\gap_f)$, where $\le_f$ is a total pre-order\footnote{a 
	reflexive and transitive relation such that for all $x,y$, either $x \le_f 
	y$ or $y \le_f x$} on the set of 
	look-ahead terms $\lat{s-1}$ and $\gap_f \colon \fbv \times [0,s-1] \to 
	\Nat$ 
	is a function such that for all $i \in [0,s-1]$, $\lambda x.\gap_f(x,i)$\footnote{Note that we could have used a function $h_i(x) = \gap_f(x,i)$ instead of using the lambda notation. But this introduces a new notation---the function $h_i$, which will not be used anywhere else.}
	is a gap function.
\end{definition}

In the notation $s$-frame, $s$ is intended to denote the size of the 
frame---the number of successive positions about which information 
is captured. The current position and the following $(s-1)$ positions are 
considered, 
for which the look-ahead terms in $\lat{s-1}$ are needed. We denote 
by $<_f$ and $\equiv_f$ the strict order and equivalence relation 
induced by $\le_f \colon x <_f y$ iff $x \le_f y$ and $y \not\le_f x$ and 
$x \equiv_f y$ iff $x \le_f y$ and $y \le_f x$.

We will deal with symbolic models that constitute sequences of frames. 
An $s$-frame will capture information about the first $s$ 
positions of a model. If this is followed by a $(s+1)$-frame, it will 
capture 
information about the first $(s+1)$ positions of the model. Both frames 
capture information about the first $s$ positions, so they must be 
consistent about the information they have about the shared positions. 
Similarly, an $s$-frame meant for positions $i$ to $i+s-1$ may be 
followed by another $s$-frame meant for positions $i+1$ to $i+s$. 
The two frames must be consistent about the positions $i+1$ to 
$i+s-1$ that they share. The following definition formalizes these 
requirements.
\begin{definition}[One-step compatibility]
	For $s \ge 1$, an $s$-frame $f$ and an $(s+1)$-frame $g$, the pair 
	$(f,g)$ is one-step compatible if the following conditions are 
	true.
	\begin{itemize}
		\item For all terms $t_1,t_2 \in \lat{s-1}$, $t_1 \le_f t_2$ iff $t_1 
		\le_g 
		t_2$.
		\item For all $j  \in [0,s-1]$ and all variables $x\in \fbv$, 
		$\gap_f(x,j) = 
		\gap_g(x,j)$.
	\end{itemize}
    For $s \ge 2$ and $s$-frames $f,g$, the pair $(f,g)$ is one-step 
    compatible if:
    \begin{itemize}
    	\item For all terms $t_1,t_2 \in \lat{s-2}$, $\nxt t_1 \le_f \nxt 
    	t_2$ iff $t_1 \le_g t_2$ and
    	\item for all $j \in [0,s-2]$ and all variables $x \in \fbv$, 
    	$\gap_f(x,j+1) = 
    	\gap_g(x,j)$.
    \end{itemize}
\end{definition}

Fix a number $k\ge 0$ and consider formulas of $\nxt$-length $k$. 
A 
symbolic model is a sequence $\rho$ of frames such that for all $i \ge 0$, 
$\rho(i)$ is an $\ceil{i+1}{k+1}$-frame and $(\rho(i),\rho(i+1))$ is 
one-step 
compatible. CLTL formulas can be interpreted on symbolic models, using symbolic semantics $\models_s$ as explained next. To check if the $i$\textsuperscript{th} position of $\rho$ symbolically satisfies the atomic constraint $t_1 < t_2$ (where $t_1,t_2$ are look-ahead terms), we check whether $t_1 < t_2$ according to the $i$\textsuperscript{th} frame $\rho(i)$. In formal notation, this is written as $\rho,i \models_s t_1 < t_2$ if $t_1 <_{\rho(i)} t_2$. For 
future-blind variables $x,y$, $\rho,i \models_s x < y$ if $\gap_{\rho(i)}(x,0) < \gap_{\rho(i)}(y,0)$. The 
symbolic satisfaction relation $\models_s$ is extended to all CLTL 
formulas of $\nxt$-length $k$ by induction on structure of the 
formula, 
as done for propositional LTL.
To check whether $\rho,i 
\models_s t_1 < t_2$ in this symbolic 
semantics, we only need to check $\rho(i)$, the $i$\textsuperscript{th} 
frame in $\rho$, unlike the CLTL semantics, where we may need to 
check other positions also. In this sense, the symbolic semantics lets us 
treat CLTL formulas as if they were formulas in propositional LTL and 
employ techniques that have been developed for propositional LTL. But 
to complete that task, we need a way to go back and forth between 
symbolic and concrete models.

Given a concrete model $\sigma$, we associate with it a symbolic 
model $\sm(\sigma)$ as follows. Imagine we are looking at the concrete model through a narrow aperture that only allows to view $k+1$ positions of the concrete model, and we can slide the aperture to view different portions. The $i$\textsuperscript{th} frame of 
$\sm(\sigma)$ will 
capture information about the portion of the concrete model visible when the right tip of the aperture is at position $i$ of the concrete model (so the left tip will be at $i - \ceil{i}{k}$). Formally, the total pre-order of the $i$\textsuperscript{th} frame is the one induced by the valuations along the positions $i-\ceil{i}{k}$ to $i$ of the concrete model. For every $j \in 
[0,\ceil{i}{k}]$, the function $\lambda x.\gap_f(x,j)$ of the $i$\textsuperscript{th} frame is the gap function associated with the mapping $\sigma(i - \ceil{i}{k}+j) \restfbv$.

For every concrete model, there is an associated symbolic model, but the converse is not true. E.g., if every frame in a symbolic model requires 
$\nxt x < x$, the corresponding concrete model needs to have an 
infinite descending chain, which is not possible in the constraint system 
$(\Nat, <, =)$. We
say that 
a symbolic model $\rho$ admits a concrete model if there 
exists a concrete model $\sigma$ such that $\rho = \sm(\sigma)$.
\begin{lemma}[{\cite[Lemma~3.1]{DD07}}]
	\label{lem:symbModelConcrMdel}
	Let $\phi$ be a CLTL formula of $\nxt$-length $k$. Let $\sigma$ 
	be a 
	concrete model and $\rho=\sm(\sigma)$. Then $\sigma,0 \models 
	\phi$ iff $\rho,k \models_s \phi$.
\end{lemma}

%% file: completion.tex
In this section, we prove that the CLTL realizability problem is 
decidable if the domain satisfies a so called completion property. Let 
$C$ be a set of constraints over a constraint system $\csys$. We call 
$C$ satisfiable if 
there is a valuation satisfying all the constraints in $C$. For a subset $U 
\subseteq V$ of variables, $C \restr U$ is the subset of $C$ 
consisting of those constraints that only use terms built with variables in 
$U$. A partial valuation $v'$ is a valuation for the terms occurring in 
$C \restr U$. We say $\csys$ has the 
completion 
property if for every satisfiable set of constraints $C$ and every subset 
$U \subseteq V$, every partial valuation $v'$ satisfying $C \restr U$ 
can be 
extended to a valuation $v$ satisfying $C$. An example of a constraint 
system which does not satisfy the completion property is $(\Int,<,=)$, 
since for the set of constraints $C = \set{x < y,x < z,z < y}$ over the set 
of variables $V = \set{x, y, z}$, the partial valuation $v \colon x \mapsto 0, y 
\mapsto 1$ satisfies the constraints in $C$ involving $x$ and $y$, but 
cannot be extended to a valuation which satisfies the constraints $x < 
z$ and $z < y$ in $C$. The constraint systems $(\mathbb{Q},<,=)$ 
and $(\mathbb{R}, <,=)$ satisfy the completion property. Also, one can easily see that for every infinite domain $D$, the constraint system $(D,=)$ always satisfies the completion property.

It is known that CLTL satisfiability is decidable for constraint systems 
that satisfy the completion property \cite{DD07,BC02}. 
The completion property of a constraint system is closely related to the 
denseness of the underlying domain. A constraint system satisfies the completion property if and only if the underlying domain is dense and open \cite[Lemma 5.3]{DD07}. Now we prove that 
for constraint 
systems of the form $(D,<,=)$ that satisfy the completion property, the 
CLTL realizability problem is decidable. This 
holds even when both players have look-ahead variables, so we don't 
need to treat future-blind variables separately. Hence, we set 
$\fbv$ to be empty and ignore gap functions in frames.

We reduce CLTL games to parity games on finite graphs, which are 
known to be decidable (see, e.g., \cite{M02}). In a CLTL game, \env chooses 
a valuation for $\ev$, which we track in our finite graph by storing the 
positions of the new values relative to the values chosen in the previous 
rounds. We do this with partial frames, which we define next.
\begin{definition}[Partial frames and compatibility]
	For $s \ge 1$, a partial $s$-frame $\pfr$ is a total pre-order 
	$\le_\pfr$ on the 
	set of terms $\lat{s-2} \cup 
	\set{\nxt^{s-1} y \mid y \in \ev}$. For $s \ge 0$, an $s$-frame $f$ 
	and an $(s+1)$-partial frame $\pfr$, the pair $(f,\pfr)$ is one step 
	compatible if for all $t_1, t_2 \in \lat{s-1}$, $t_1 \le_f 
	t_2$ iff $t_1 \le_\pfr t_2$. For $s \ge 2$, an 
	$s$-frame $f$ and an $s$-partial frame $\pfr$, the pair $(f,\pfr)$ is 
	one-step compatible if for all $t_1, t_2 \in \lat{s-2}$, 
	$\nxt t_1 \le_f \nxt t_2$ iff $t_1 \le_\pfr t_2$. For 
	$s 
	\ge 2$, an $s$-partial frame $\pfr$ and an $s$-frame $f$, $(\pfr,f)$ 
	is one step compatible if for all $t_1, t_2 \in \lat{s-2} \cup \set{\nxt^{s-1} y \mid y \in \ev}$, 
	$t_1 \le_\pfr t_2$ iff  $t_1 \le_f t_2$.
\end{definition}
In the set of terms $\lat{s-2} \cup 
\set{\nxt^{s-1} y \mid y \in \ev}$ used in partial frames, the terms in 
the first set represent values chosen in the previous rounds and the 
terms in the second set represent values chosen by \env for $\ev$ in 
the current round.

Note that a partial $s$-frame is a total pre-order on the set of terms 
$\lat{s-2} \cup \set{\nxt^{s-1} y \mid y \in \ev}$ and an $s$-frame is 
a total pre-order on the set of terms $\lat{s-1}$. Let $pf$ be an 
$s$-partial frame and let $f$ be an $s$-frame such that $(pf,f)$ is 
one-step compatible. Suppose $C_1 = \set{t_1 = t_2 \mid t_1 
\equiv_{f} t_2} \cup \set{t_1 < t_2 \mid t_1 <_{f} t_2}$ and $C_2 = 
\set{t_1 = t_2 \mid t_1 \equiv_{pf} t_2 } \cup \set{t_1 < t_2 \mid t_1 
<_{pf} t_2}$. Clearly, $C_2$ is a subset of $C_1$ skipping all those 
constraints that contain \sys variables corresponding to the 
$s^{\text{th}}$ position. If a finite sequence of mappings $({\emap}_1 \oplus 
{\smap}_1)...({\emap}_{s-1} \oplus {\smap}_{s-1}){\emap}_{s}$ satisfies the pre-order $\le_\pfr$ then it satisfies the 
constraints in $C_2$. Since the constraint system satisfies the 
completion property, there must exist a \sys mapping 
${\smap}_{s}$ for the \sys variables at position $s$ such that 
the sequence of mappings $({\emap}_1 \oplus {\smap}_1)...({\emap}_{s} 
\oplus {\smap}_{s})$ satisfies the constraints in $C_1$ and hence, also satisfies
the pre-order $\le_f$. Thus, we have the following proposition:

\begin{proposition}\label{prop:partial_frames_and_completion}
Given $s \ge 1$, suppose $({\emap}_1 \oplus {\smap}_1)...({\emap}_i \oplus {\smap}_i){\emap}$ 
is a sequence of mappings, where ${\emap}_1,\ldots,{\emap}_i,\emap \in \eMap$, ${\smap}_1,\ldots,{\smap}_i \in \sMap$, $pf$ is the $s$-partial frame induced by 
${\emap}$ and the previous $(s-1)$ mappings in the sequence, and $f$ 
is an $s$-frame such that $(pf,f)$ is one-step compatible (where $i \ge 
s$). If the constraint system satisfies the completion property, then ${\emap}$ can be extended to a mapping ${\emap} \oplus {\smap}$ 
such that $f$ is the $s$-frame associated with ${\emap} \oplus 
{\smap}$ and the previous $(s-1)$ mappings in the sequence.
\end{proposition}
We know that any LTL formula $\phi$ can be converted to an equivalent non-deterministic B\"{u}chi automaton with an exponential number of states in the size of $\phi$ in \textsc{EXPTIME} \cite{VW86}. Now, every non-deterministic B\"{u}chi automaton $B$ with $n$ states can be converted to a deterministic parity automaton \cite[Chapter 1]{GT02} with number of states exponential in $n$ and number of colours polynomial in $n$  \cite[Theorem 3.10]{P07}. Using these results, it is easy to see that given a CLTL formula $\phi$, we can construct a deterministic parity automaton $A_\phi$ with set of states $Q$ and with number of colours $d$, accepting the set of all sequences of frames that symbolically satisfy $\phi$, such that $|Q|$ is double exponential in the size of $\phi$ and $d$ is exponential in the size of $\phi$. Now we design parity games to simulate CLTL games.
\begin{definition}
	\label{def:parityGame}
	Let $\phi$ be the CLTL formula defining the winning condition for a 
	CLTL game and $k$ be its $\nxt$-length. Let $\frames$ denote the set of all $s$-frames for $s\in [0,k]$. Let $A_\phi$ be a 
deterministic parity automaton accepting the set of all sequences of 
	frames that symbolically satisfy $\phi$, with $Q$ being the set of 
	states, $q_I\in Q$ being the initial state and $d$ being the number of colours. We define a parity game 
	with \env 
	vertices $V_e = \set{(f,q_I) \mid f \text{ is an } s\text{-frame}, 0 \le s 
	\le k} 
	\cup \set{(f,q) \mid f \text{ is a } (k+1)\text{-frame}, q \in Q}$. The 
	set of \sys vertices is $V_s = \set{(f,q_I,\pfr) \mid f \text{ is an } 
	s\text{-frame}, 0 \le s \le k, \pfr \text{ is an }(s+1)\text{-partial 
	frame}} \cup \set{(f,q,\pfr) \mid f \text{ is a } (k+1)\text{-frame}, 
	\pfr \text{ is a }(k+1)\text{-partial frame}}$. There is an edge from 
	$(f,q)$ to $(f,q,\pfr)$ if $(f,\pfr)$ is one-step compatible, $f$ is an 
	$s$-frame for some $s$ and $\pfr$ is a partial 
	$\ceil{s+1}{(k+1)}$-frame. There is 
	an 
	edge from $(f,q_I,\pfr)$ to $(g,q_I)$ if $(\pfr,g)$ is one step 
	compatible and $g$ is an $s$-frame for $s \in [1,k]$. There is an 
	edge from $(f,q,\pfr)$ to $(g,q')$ if $(\pfr,g)$ is one-step 
	compatible, 
	$g$ is a $(k+1)$-frame and $A_\phi$ goes from $q$ to $q'$ on 
	reading $g$. Vertices $(f,q)$ and $(f,q,\pfr)$ get the same colour as 
	$q$ in the parity automaton $A_\phi$. The initial vertex is 
	$(\bot,q_I)$, where $\bot$ is the trivial $0$-frame.
\end{definition}
The edges of the parity game above are from $V_s$ to $V_e$ or 
vice-versa. They are designed such that $q_I$ is the 
only state used for the first $k$ rounds, where the frames will be of 
size at 
most $k$ (this is because for the \sys to win in a play of the parity game generating a frame sequence $\rho$, we only require that the sequence $\boldsymbol{\rho[k,\infty)}$ symbolically satisfy $\phi$, according to Lemma~\ref{lem:symbModelConcrMdel}). For the first $(k+1)$ frame, an edge from a \sys vertex of 
the form $(f,q_I,\pfr)$ to an \env vertex of the form $(g,q')$ is taken 
and from then on, we track the state of the parity automaton as it reads 
the sequence of frames contained in the sequence of vertices that are 
chosen by the players in the game.
\begin{lemma}
	\label{lem:cltlToParityGame}
	For a CLTL game over a constraint system satisfying the completion 
	property with winning condition given by a formula $\phi$, \sys has 
	a winning strategy iff she has a positional winning strategy in the 
	parity game given in Definition~\ref{def:parityGame}.
\end{lemma}
\begin{proof}[Proof idea]
	For every play in the CLTL game, there is a corresponding play in the 
	parity game, but the converse is not true in general, since only the 
	order of terms are tracked in the parity game and not the actual 
	values. For constraint systems satisfying the completion property, 
	Proposition~\ref{prop:partial_frames_and_completion} implies that 
	there exist valuations corresponding to all possible orderings 
	of terms, so the converse is also true.
\end{proof}

\begin{proof}
	($\Rightarrow$) Suppose \sys has a winning strategy $\sstrat$ in the 
	CLTL game. We show that \sys has a winning strategy in the 
	parity game. Plays in the parity game are of the form 
	$(\bot, q_I) (\bot, 
	q_I, \pfr_1) (f_1,q_1) (f_1,q_1,\pfr_2) (f_2,q_2) \cdots 
	(f_i,q_i,\pfr_{i+1}) (f_{i+1},q_{i+1}) \cdots$, where $(f_i,\pfr_{i+1})$ 
	and $(\pfr_{i+1},f_{i+i})$ are one-step compatible for all $i$. For 
	any such play $\pi$, let $\pi \restr i$ be $(\bot, q_I) (\bot, 
	q_I, \pfr_1) (f_1,q_1) (f_1,q_1,\pfr_2) (f_2,q_2) \cdots 
	(f_i,q_i,\pfr_{i+1})$. Let $\Pi = \set{\pi \restr i \mid \pi \text{ is a 
	play in the parity game}, i \ge 0}$. We will show the existence of a 
	function $\sstrat_p \colon \Pi \to V_e \times \eMap \times \sMap$ 
	satisfying some properties. Such a function can be used as a strategy 
	by \sys in the parity game: for a play $\pi\restr i$, \sys's response 
	$(f_{i+1},q_{i+1})$ is given by $\sstrat_p$, i.e., $\sstrat_p(\pi\restr i) 
	= ((f_{i+1},q_{i+1}),\emap_{i+1},\smap_{i+1})$. For such plays that 
	\sys plays according $\sstrat_p$, let $\frs(\pi \restr i)$ be the 
	symbolic model 
	$f_1 f_2 \cdots f_{i+1}$ and let $\maps(\pi \restr i)$ be the 
	concrete model 
	$(\emap_1 \oplus \smap_1) (\emap_2 \oplus \smap_2) \cdots 
	(\emap_{i+1} \oplus \smap_{i+1})$.
	
	We will show that there is a 
	function $\sstrat_p$ such that for all plays $\pi$ that 
	\sys 
	plays according to $\sstrat_p$ and all $i \ge 0$ , $\maps(\pi \restr 
	i)$ 
	is a 
	concrete model resulting from \sys playing the CLTL game according 
	to $\sstrat$ and $\frs(\pi \restr i) = \sm(\maps(\pi \restr i))$. 
	We will define such a function $\sstrat_p$ by induction on $i$. We 
	assume this has been done for $i$ and show how to extend to 
	$i+1$. We have $\pi \restr (i+1)=(\bot, q_I) (\bot, 
	q_I, \pfr_1) (f_1,q_1) (f_1,q_1,\pfr_2) (f_2,q_2) \cdots 
	(f_i,q_i,\pfr_{i+1}) (f_{i+1}, q_{i+1}) \\(f_{i+1}, q_{i+1},\pfr_{i+2})$. 
	By 
	induction hypothesis, $f_1 f_2 \cdots f_{i+1}=\sm(\maps(\pi \restr 
	i))$. Since the constraint system satisfies the completion property 
	and $(f_{i+1},\pfr_{i+2})$ is one-step compatible, by 
	Proposition~\ref{prop:partial_frames_and_completion}, there is a 
	mapping 
	$\emap \colon \ev \to D$ such that the symbolic model induced by 
	$\maps(\pi \restr i)\cdot \emap$ is $f_1 f_2 \cdots f_{i+1} \cdot 
	\pfr_{i+2}$. Let $\smap \colon \sv \to D = \sstrat(\maps(\pi \restr i)\cdot 
	\emap)$ be \sys's response in the CLTL game according to 
	$\sstrat$. Let $f_{i+2}$ be the frame such that $f_1 f_2 \cdots 
	f_{i+1}f_{i+2} = \sm(\maps(\pi \restr i)\cdot (\emap \oplus 
	\smap))$. Set $\sstrat_p(\pi \restr (i+1))$ to be 
	$((f_{i+2},q_{i+2}),\emap,\smap)$, where $q_{i+2}$ is the state 
	$A_\phi$ reaches after reading $f_{i+2}$ in state $q_{i+1}$. Now, 
	$\maps(\pi \restr (i+1))$ 
	is a 
	concrete model resulting from \sys playing the CLTL game according 
	to $\sstrat$ and $\frs(\pi \restr (i+1)) = \sm(\maps(\pi \restr 
	(i+1)))$, as required for the inductive construction.
	
	Let $\pi$ be any infinite play in the parity game that \sys plays 
	according to $\sstrat_p$. Then $\maps(\pi)$ is a concrete model 
	resulting from \sys playing the CLTL game according to $\sstrat$ 
	and $\frs(\pi) = \sm(\maps(\pi))$. Since $\sstrat$ is a winning 
	strategy for \sys, $\maps(\pi),0 \models \phi$. We infer from 
	Lemma~\ref{lem:symbModelConcrMdel} that $\frs(\maps(\pi)),k 
	\models_s \phi$. Hence, the sequence of states $q_I, q_1, q_2, 
	\ldots$ contained in the sequence of vertices that are visited in $\pi$ 
	satisfy the parity condition of $A_\phi$. Hence, $\pi$ itself satisfies 
	the parity condition and hence \sys wins $\pi$. Hence, $\sstrat_p$ 
	is a winning strategy for \sys in the parity game.
	
	($\Leftarrow$) Suppose $\sstrat_p$ is a positional strategy for \sys 
	in the parity game. We will show that \sys has a winning strategy 
	$\sstrat$ in the CLTL game. We will define $\sstrat$ by induction on 
	the number of rounds played. For the base case, suppose \env starts 
	by choosing a mapping $\emap_1 \colon \ev \to D$. In the parity game, let 
	\env go to the vertex $(\bot,q_I, \pfr_1)$ in the first round, where 
	$\pfr_1$ is the 
	$1$-partial frame associated with $\emap_1$. Let 
	$(f_1,q_1)=\sstrat_p((\bot,q_I,\pfr_1))$ be \sys's 
	response according to $\sstrat_p$. Since $(\pfr_1,f_1)$ is one-step 
	compatible and the constraint system satisfies the completion 
	property, by Proposition~\ref{prop:partial_frames_and_completion}, 
	$\emap_1$ can be extended to a mapping $\emap_1 
	\oplus 
	\smap_1 \colon V \to D$ such that $f_1$ is the frame associated with 
	$\emap_1 \oplus \smap_1$. Set $\sstrat(\emap_1)$ to be 
	$\smap_1$. After $i$ rounds of the CLTL game, suppose $(\emap_1 
	\oplus \smap_1) \cdots (\emap_i \oplus \smap_i)$ is the resulting 
	concrete model and let $(\bot, q_I) (\bot, 
	q_I, \pfr_1) (f_1,q_1) \cdots (f_i,q_i)$ be the corresponding play in 
	the parity game. Suppose \env chooses $\emap_{i+1}$ in the next 
	round. Let $\pfr_{i+1},f_{i+1},q_{i+1}, \smap_{i+1}$ be obtained 
	similarly as in the base case. Set $\sstrat((\emap_1 
	\oplus \smap_1) \cdots (\emap_i \oplus \smap_i)\cdot 
	\emap_{i+1})$ to be $\smap_{i+1}$.
	
	Suppose $(\emap_1 \oplus \smap_1)(\emap_2 \oplus \smap_2) 
	\cdots$ is an infinite play in the CLTL game that \sys plays according 
	to $\sstrat$. There is a play $(\bot, q_I) (\bot, 
	q_I, \pfr_1) (f_1,q_1) (f_1,q_1,\pfr_2) (f_2,q_2) \cdots $ in the parity 
	game that is winning for \sys. This satisfies the parity condition, 
	hence $A_\phi$ accepts the symbolic model $f_1 f_2 \cdots$. The 
	symbolic model $f_1 f_2 \cdots $ is the one associated with 
	$(\emap_1 \oplus \smap_1)(\emap_2 \oplus \smap_2) \cdots$ by 
	construction of $\sstrat$, so Lemma~\ref{lem:symbModelConcrMdel} 
	implies that $(\emap_1 \oplus \smap_1)(\emap_2 \oplus \smap_2) 
	\cdots,0 \models \phi$. Hence, $\sstrat$ is a winning strategy for 
	\sys in the CLTL game.
\end{proof}

\begin{theorem}
\label{thm:deccompletion}
	The CLTL realizability problem over 
	constraint systems that satisfy the completion property is \textsc{2EXPTIME}-complete.
\end{theorem}
\begin{proof}
	From Lemma~\ref{lem:cltlToParityGame}, this is effectively equivalent 
	to checking the existence of a winning strategy for \sys in a 
	game. Now, checking if \sys has a winning strategy in the parity game (constructed using $A_\phi$) can be achieved in $O(n^{\log{d}})$ time where $n$ is the number of states in the game graph \cite{CJKLS20}. Now, by our construction, $n = |Q| \times |\frames|$. We know, $|\frames|$ is the number of total pre-orders on $V$, for which $2^{{(k.|V|)}^2}$ is a crude upper bound. This means that $|\frames|$ is exponential in the size of $\phi$ and hence, overall we get a \textsc{2EXPTIME} upper bound for our realizability problem.
	We also know that the realizability problem for LTL is complete for \textsc{2EXPTIME} \cite{PR89} and every LTL formula is also a CLTL formula. Thus, the CLTL realizability problem over constraint systems satisfying the completion property is also \textsc{2EXPTIME}-complete.
\end{proof}
We know that a positional winning strategy in the parity game for a player, if it exists, can be implemented by a deterministic finite state transducer.  Since $\csys$ satisfies the completion property, consider a resource-bounded Turing machine $M$, which can, given an environment mapping $\emap$ as described in Proposition 7, extend it to a mapping $\emap \oplus \smap$ such that the order $f$ imposed by the $\emap \oplus \smap$ and the previous $s-1$ mappings over the set of all terms extends the order $pf$ imposed by $em$ and the previous $s-1$ mappings. Now, for implementing the winning strategy for a player in a CLTL game, we use the deterministic finite state transducer  corresponding to the parity game given in Definition~\ref{def:parityGame}. For every input of a partial frame $pf$ by \env in a round, the transducer returns a frame $f$ for \sys that extends $pf$. The transducer along with the machine $M$ implements the winning strategy for \sys in a given CLTL game, if it exists.

Note that as we saw above, the constraint systems $(\Nat,=)$ and $(\Int,=)$ (with just equality and no linear order) also satisfy the completion property. So, it follows that the CLTL realizability problem over these constraint systems is also decidable.

%% file: singleSided.tex
We consider games where \env has only future-blind 
variables, while the \sys has both future-blind and look-ahead 
variables. We call this single-sided CLTL games. So, in a 
single-sided game, $\ev = \efbv$ and $\sv = \sfbv \cup \slav$. 
Given a CLTL formula $\phi$, the \textbf{single-sided realizability problem} is to check whether \sys has a winning 
strategy in the single-sided CLTL game whose winning condition is 
$\phi$. We only consider the constraint system 
$(\Int,<,=)$ and show that the single-sided realizability problem is 
decidable over $(\Int,<,=)$. We do this in two stages. In the first stage, we reduce it to 
the problem of checking the non-emptiness of a set of trees satisfying 
certain properties. These trees represent \sys strategies. In the second 
stage, we show that non-emptiness can be checked using tree 
automata techniques.

Let $\gf$ be the set of gap 
functions associated with mappings of the form $\efbv \to \Int$.
For $s \ge 1$, an $s$-frame $g$ and a 
function $\gap \in \gf$, the pair $(\gap,g)$ is gap compatible 
if for all $x,y \in \efbv$, 
$\gap(x) - \gap(y) = \gap_g(x,s-1) - \gap_g(y,s-1)$. 
Intuitively, the gaps  
that frame $g$ imposes between $\efbv$ variables in its last position 
are same as the gaps imposed by $\gap$.
\begin{proposition}[gap compatibility]
	\label{prop:gapCompatible}
	For $s \ge 1$, an $s$-frame $g$ and a function $\gap \in 
	\gf$, 
	suppose the pair $(\gap,g)$ is gap compatible. If $\gap$ is the gap 
	function associated with a mapping $\emap \colon \efbv \to \Int$, it can 
	be extended to a mapping $\emap \oplus \smap \colon \fbv \to \Int$ 
	such that 
	$\lambda x . \gap_g(x,s-1)$ is the gap function associated 
	with 
	$\emap \oplus \smap$.
\end{proposition}

\begin{proof}
	Arrange $\efbv$ as $x_0, x_1, \ldots$ such that $0=\gap(x_0) < 
	\gap(x_1) < \cdots$. For any $j$, let $V_j = \set{y \in \sfbv \mid 
	\gap_g(x_{j-1},s-1) < \gap_g(y,s-1) < \gap_g(x_{j},s-1)}$ be the set 
	of variables in $\sfbv$ that occur ``in-between'' $x_{j-1}$ and $x_j$ 
	according to $\gap_g(\cdot,s-1)$. Since the pair $(\gap,g)$ is gap 
	compatible, $\gap_g(x_j,s-1) - \gap_g(x_{j-1},s-1)=\gap(x_j) - 
	\gap(x_{j-1})=\ceil{\emap(x_j) - \emap(x_{j-1})}{|\fbv|-1}$. The 
	second equality follows from the fact that $\gap$ is the gap function 
	associated with $\emap$. Since $|V_j| < |\fbv|-1$, the gap 
	$\ceil{\emap(x_j) - \emap(x_{j-1})}{|\fbv|-1}$ is wide enough to 
	accommodate all 
	variables in $V_j$. The mapping $\smap$
	assigns to each variable in $V_j$ some value between 
	$\emap(x_{j-1})$ and $\emap(x_{j})$ such that $\lambda x . 
	\gap_g(x,s-1)$ is the gap function associated with $\emap \oplus 
	\smap$.
\end{proof}

Let $\phi$ be the CLTL formula defining the winning condition of a 
single-sided CLTL game and let $k$ be its 
$\nxt$-length.  
Let $\frames$ be 
the set of all $s$-frames for $s\in [0,k]$. For technical convenience, we 
let $\frames$ include the trivial $0$-frame 
$\bot=(\le_\bot,\gap_\bot)$, where $\le_\bot$ is the trivial total 
pre-order on the empty set and $\gap_\bot$ is the trivial function on 
the empty domain.
\begin{definition}[Winning strategy trees]
	A strategy tree is a function $\tree \colon {\gf}^* \to \frames$ such that for 
	every 
	node $\node \in {\gf}^*$, $\tree(\node)$ is a 
	$\ceil{|\node|}{k+1}$-frame and for every $\gap \in \gf$, 
	$(\tree(\node),\tree(\node \cdot \gap))$ is one-step 
	compatible and $(\gap,\tree(\node \cdot \gap))$ is gap compatible.
	A function $\maplabel$ is said to be a labeling function if for every node 
	$\node \in 
	{\gf}^*$, $\maplabel(\node) \colon V \to \Int$ is a mapping of the variables in 
	$V$.
	For an infinite path $\pth$ in $\tree$, let $\tree(\pth)$ 
	(resp.~$\maplabel(\pth)$) denote the infinite sequence of frames 
	(resp.~mappings) labeling the nodes in $\pth$, except the root node 
	$\epsilon$. A winning strategy tree is a pair 
	$(\tree,\maplabel)$ such that $\tree$ is a strategy tree and $\maplabel$ is a labelling function satisfying the condition that for every infinite path 
	$\pth$, $\tree(\pth) = \sm(\maplabel(\pth))$ and $\tree(\pth),k 
	\models_s \phi$.
\end{definition}
The last condition above means that $\tree(\pth)$ is the symbolic 
model associated with the concrete model $\maplabel(\pth)$ and that 
it symbolically satisfies the formula $\phi$.

Two concrete models may have the same symbolic model associated 
with them, if they differ only slightly, as explained next. Two concrete 
models $\sigma_1,\sigma_2$ are said to coincide on $\lav$ if 
$\sigma_1(i)\restr \lav = 
\sigma_2(i) \restr \lav$ for all $i \ge 0$. They are said to coincide on 
$\fbv$ up to gap functions if for every $i \ge 0$, the same gap function 
is associated with $\sigma_1(i)\restr \fbv$ and $\sigma_2(i)\restr 
\fbv$. The following result follows directly from definitions.
\begin{proposition}[similar concrete models have same symbolic model]
	\label{prop:similarConcrModels}
	If two concrete models coincide on $\lav$ and they coincide on 
	$\fbv$ up to gap functions, then they have the same symbolic model 
	associated with them.
\end{proposition}

The following result accomplishes the first stage of the decidability 
proof, reducing the existence of winning strategies to 
non-emptiness of a set of trees.
\begin{lemma}[strategy to tree]
	\label{lem:stratToTree}
	\Sys has a winning strategy in the single-sided CLTL game with 
	wining condition $\phi$ iff there exists a winning strategy tree.
\end{lemma}

\begin{proof}[Proof idea]
	If \env chooses a mapping $\emap \colon \efbv \to \Int$ in the CLTL 
	game, the corresponding choice in the tree $\tree$ is to go to the 
	child $\gap$, the gap function associated with $\emap$. \Sys 
	responds with the mapping $\maplabel(\gap)\restr \slav$ for the 
	look-ahead variables. For the future-blind variables $\sfbv$, \sys 
	chooses a mapping that ensures compatibility with the frame 
	$\tree(\gap)$. This will ensure that \sys's response and $L$ coincide 
	on $\lav$ and coincide on $\fbv$ up to gap functions, so 
	Proposition~\ref{prop:similarConcrModels} ensures that both have 
	the 
	same symbolic model. The symbolic model symbolically satisfies 
	$\phi$ by definition of wining strategy trees and 
	Lemma~\ref{lem:symbModelConcrMdel} implies that 
	the concrete model satisfies $\phi$.
\end{proof}

\begin{proof}
	We introduce a notation to use in this proof. For a labeling function 
	$\maplabel$ and a node $\node = \gap_1 \cdot \gap_2 \cdots 
	\gap_n$, we denote by $\hat{\maplabel}(\node)$ the 
	sequence of mappings $\maplabel(\gap_1) \cdot \maplabel(\gap_1 
	\cdot \gap_2) \cdots \maplabel(\node)$. The notation $\hat{\tree}$ 
	has similar meaning.
	
	($\Rightarrow$) Suppose \sys has a winning strategy $\sstrat$. We 
	first construct a labeling function $\maplabel$ such that for 
	any node $\node$ with $|\node|\ge 1$, if 
	$\hat{\maplabel}(\node)\restefbv$ is the sequence of \env 
	choices, then
	$\hat{\maplabel}(\node)\restsv$ are the \sys responses 
	according to 
	$\sstrat$. We proceed by induction on $|\node|$. For the base case 
	$|\node| = 0$, $\node=\epsilon$. We set 
	$\maplabel(\epsilon)$ to be the constant function that maps 
	everything to $0$. This satisfies the specified condition vacuously.
	
	For the induction step, consider a node $\node\cdot 
	\gap$ for some $\gap \in \gf$. Let $\emap_{\gap} \colon \efbv \to \Int$ 
	be a mapping such that the gap function associated with 
	$\emap_{\gap}$ is $\gap$. Set $\maplabel(\node\cdot 
	\gap)=\emap_\gap \oplus \sstrat(\hat{\maplabel}(\node)\cdot 
	\emap_\gap)$. By induction hypothesis, for the 
	sequence $\hat{\maplabel}(\node)\restefbv$ of \env choices, 
	$\hat{\maplabel}(\node)\restsv$ are the \sys responses 
	according to $\sstrat$. Hence, for the sequence 
	$\hat{\maplabel}(\node)\restefbv \cdot \emap_\gap$ of \env 
	choices, $\hat{\maplabel}(\node)\restsv \cdot 
	\sstrat(\hat{\maplabel}(\node)\cdot \emap_\gap)$ are 
	the \sys responses. This completes the induction step and the 
	construction of the labeling 
	function $\maplabel$. Let $\tree$ be the tree such that for every 
	infinite 
	path $\pth$, $\tree(\pth)$ is $\sm(\maplabel(\pth))$, the symbolic 
	model associated with the concrete model $\maplabel(\pth)$. Since 
	$\maplabel(\pth)\restsv$ are the \sys responses to \env choices 
	$\maplabel(\pth)\restefbv$ according to the winning strategy 
	$\sstrat$, 
	$\maplabel(\pth),0 \models \phi$. 
	Lemma~\ref{lem:symbModelConcrMdel} implies that 
	$\sm(\maplabel(\pth)),k \models_s \phi$, so $\tree(\pth),k
	\models_s \phi$.
	
	($\Leftarrow$) For a sequence $\sigma$ of mappings over $V$ and 
	$S\subseteq V$, let 
	$\hat{\gap}(\sigma\restr S)$ denote the sequence of gap functions 
	$\gap_1 \cdot \gap_2 \cdots$ such that for all $i$, $\gap_i$ is the 
	gap function associated with the mapping $\sigma(i)$ restricted to 
	the domain $S$.
	Suppose there exists a winning strategy tree 
	$(\tree,\maplabel)$. We will construct a strategy $\sstrat$ that is 
	winning for \sys satisfying the following property: suppose 
	$\sigma\in 
	\Map^*$ is a sequence of mappings resulting 
	from a play that \sys plays according to $\sstrat$ and 
	$\node=\hat{\gap}(\sigma\restefbv)$. Then 
	$\hat{\maplabel}(\node)$ and 
	$\sigma$ coincide on $\slav$ and they coincide on $\fbv$ 
	up to gap functions. We proceed by induction 
	on 
	$|\sigma|$. For the base case $|\sigma|=0$, $\sigma=\epsilon$ 
	and there is nothing to prove.
	
	For the induction step, let $\sigma$ be the sequence of mappings 
	resulting from the rounds played so far and in the next round, 
	suppose \env chooses the mapping $\emap \colon \efbv \to \Int$. Let 
	$\gap_\emap$ be 
	the gap function associated with $\emap$, $\node = 
	\hat{\gap}(\sigma\restefbv)$ and $f=\tree(\node\cdot 
	\gap_\emap)$. 
	Say $f$ is an $s$-frame for some $s$. By definition of 
	winning 
	strategy trees, the pair $(\gap_\emap,f)$ is gap compatible. 
	Proposition~\ref{prop:gapCompatible} implies that $\emap$ can be 
	extended to a mapping $\emap' \colon \fbv \to \Int$ such that $\lambda 
	x. \gap_f(x,s-1)$ is the gap function associated with $\emap'$.  Also 
	$\lambda 
	x. \gap_f(x,s-1)$ is the gap function associated with 
	$\maplabel(\node 
	\cdot \gap_\emap)\restfbv$, by definition of winning strategy trees. 
	So $\emap'$ and  $\maplabel(\node 
	\cdot \gap_\emap)\restfbv$ coincide up to gap functions.
	Set 
	$\sstrat(\sigma \cdot \emap)$ to be 
	$\smap = (\maplabel(\node.\gap_\emap)\restslav )
	\oplus 
	(\emap'\restsfbv)$. In words, \sys's response resembles 
	$\maplabel(\node\cdot \gap_\emap)$ on $\slav$ and resembles 
	$\emap'$ on 
	$\sfbv$. 
	By induction hypothesis, $\hat{\maplabel}(\node)$ and 
	$\sigma$ coincide on $\slav$ and they coincide on $\fbv$ 
	up to gap functions. Hence, $\hat{\maplabel}(\node\cdot 
	\gap_\emap)$ and 
	$\sigma\cdot \smap$ coincide on $\slav$ and they coincide on 
	$\fbv$	up to gap functions. This completes the induction step and 
	hence the construction of $\sstrat$.
	
	It remains to prove that $\sstrat$ is a winning strategy. Let $\sigma$ 
	be a concrete model resulting from a play in which \sys follows the 
	strategy $\sstrat$. The sequence of gap functions 
	$\hat{\gap}(\sigma \restefbv)$ induces an infinite path $\pth$ in 
	the tree $(\tree,\maplabel)$. By construction of $\sstrat$, $\sigma$ 
	and $\maplabel(\pth)$ coincide on $\lav=\slav$ and coincide on 
	$\fbv$ up to gap functions. 
	Proposition~\ref{prop:similarConcrModels} implies that $\sigma$ 
	and 
	$\maplabel(\pth)$ have the same symbolic model $\tree(\pth)$.  By 
	definition of winning strategy trees, $\tree(\pth),k \models_s 
	\phi$. 
	Lemma~\ref{lem:symbModelConcrMdel} implies that $\sigma,0 
	\models \phi$. Since this holds for any $\sigma$ resulting from a 
	play in which \sys follows the strategy $\sstrat$, this shows that 
	$\sstrat$ is winning for \sys.
\end{proof}

Given a tree ${\gf}^* \to \frames$, a tree automaton over finite alphabets 
can check whether it is a strategy tree or not, by allowing transitions 
only between one-step and gap compatible frames. However, to check 
whether it is a winning strategy tree, we need to check whether there 
exists a labeling function $\maplabel$, which is harder. One way to 
check the existence of such a labeling function is to start labeling at the 
root and inductively extend to children. Suppose there are two variables 
$x,y$ at some node and we have to label them with integers. There 
may be many variables in other nodes whose labels should be 
strictly 
between those of $x,y$ in the current node. So our labels for $x,y$ in 
the current 
node should leave gap large enough to accommodate others that are 
supposed to be in between. Next we introduce some orderings we use 
to formalize this.

A node variable in a strategy tree $\tree$ is a pair $(\node,x)$ where 
$\node$ is a node and $x \in \lav$ is a look-ahead variable. The tree 
induces an order on node variables as follows. Suppose $\node$ is a 
node, $\tree(\node)$ is an $s$-frame for some $s$ and $\node_a$ is 
an 
ancestor of $\node$ such that the difference in height $h=|\node| - 
|\node_a|$  between the descendant and ancestor is at most $s-1$. For 
look-ahead variables $x,y \in \lav$, recall that the term $\nxt^{s-1} 
x$ represents the variable $x$ in the last position of the frame 
$\tree(\node)$, and $\nxt^{s-1-h}y$ represents the variable $y$ at 
$h$ positions before the last one. We say $(\node,x) \lo_\tree 
(\node_a,y)$ (resp.~$(\node_a,y) \lo_\tree (\node,x)$) if $\nxt^{s-1} 
x \le_{\tree(\node)} \nxt^{s-1-h}y$ (resp.~$\nxt^{s-1-h}y 
\le_{\tree(\node)} \nxt^{s-1} x$). In other words, for the variables and 
positions captured in the frame $\tree(\node)$, $\lo_\tree$ is same as 
the total pre-order $\le_{\tree(\node)}$. We define $(\node,x) \lto_\tree 
(\node_a,y)$ (resp.~$(\node_a,y) \lto_\tree (\node,x)$) if $(\node,x) \lo_\tree (\node_a,y)$ and $(\node_a,y) 
\not\lo_\tree (\node,x)$ (resp.~$(\node_a,y) 
\lo_\tree (\node,x)$ and $(\node,x) \not\lo_\tree (\node_a,y)$). We define $\ro_\tree$ to be 
the reflexive transitive closure of $\lto_\tree$ and $\rto_\tree$ to be 
the transitive closure of $\lto_\tree$. Note that $\ro_\tree$ and $\rto_\tree$ can compare 
node variables that are in different branches of the tree also, though they are not total orders. We write 
$(\node_1,x) \ro_\tree 
(\node_2,y)$ (resp,~$(\node_1,x) \rto_\tree (\node_2,y)$) equivalently as $(\node_2,y) \go_\tree (\node_1,x)$ (resp.~$(\node_1,x) \gto_\tree (\node_2,y)$). By definition,
$(\node_1,x) \rto_\tree (\node_2,y)$ (resp.$(\node_2,y) \rto_\tree 
(\node_1,x)$) if $(\node_1,x) \ro_\tree (\node_2,y)$ and $(\node_2,y) 
\not\ro_\tree (\node_1,x)$ (resp.~$(\node_2,y) 
\ro_\tree (\node_1,x)$ and $(\node_1,x) \not\ro_\tree (\node_2,y)$). 
$\rto$ is irreflexive and transitive.
\begin{definition}[Bounded chain strategy trees]
	\label{def:bddGapStrTree}
	Suppose $\tree$ is a strategy tree, $\node, \node'$ are two 
	nodes and $x,y\ \in \lav$ are 
	look-ahead variables such that $(\node,x) \rto_\tree (\node',y)$. 
	A chain between $(\node,x)$ and $(\node',y)$ is a sequence 
	$(\node_1,x_1) (\node_2,x_2) \cdots (\node_r,x_r)$ such that 
	$(\node,x) \rto_\tree (\node_1,x_1) \rto_\tree (\node_2,x_2) 
	\rto_\tree \cdots \rto_\tree (\node_r,x_r) \rto_\tree (\node',y)$. We 
	say $r$ is the length of the chain. The strategy tree $\tree$ is said to 
	have bounded chains if for any two node 
	variables $(\node,x)$ and $(\node',y)$, there is a 
	bound $N$ such that any chain between $(\node,x)$ and 
	$(\node',y)$ is of length at most $N$.
\end{definition}

\begin{lemma}
	\label{lem:bddChainAdmitsValuation}
	A strategy tree $\tree$ has a labeling function $\maplabel$ such that 
	$(\tree,\maplabel)$ is a winning strategy tree iff $\tree$ has 
	bounded chains.
\end{lemma}

\begin{proof}
    ($\Rightarrow$) Suppose $\tree$ has a labeling function 
	$\maplabel$ such that 
	$(\tree,\maplabel)$ is a winning strategy tree. Since for every infinite 
	path $\pth$, $\tree(\pth) = \sm(\maplabel(\pth))$, $\maplabel$ 
	should respect the relation $\rto_\tree$, i.e., if $(\node,x) \rto_\tree 
	(\node',y)$, then $\maplabel(\node)(x) < \maplabel(\node')(y)$. 
	Hence, any chain between $(\node,x)$ and $(\node',y)$ cannot be 
	longer than $\maplabel(\node')(y) - \maplabel(\node)(x)$.
	
	($\Leftarrow$) Suppose $\tree$ has bounded chains. We construct a 
	labeling function $\maplabel$ such that $(\tree,\maplabel)$ is a 
	winning strategy tree. At every node $\node$, we choose mappings 
	for future-blind variables $\fbv$ such that the gap function 
	associated with $\maplabel(\node)\restfbv$ is 
	$\gap_{\tree(\node)}$. These choices can be done independently for 
	every node. For look-ahead variables, we construct $\maplabel$ for 
	every node by induction on depth of the node such that for any 
	node variables $(\node,x),(\node',y)$ such that $(\node,x) \rto_\tree 
	(\node',y)$ and $\maplabel(\node), \maplabel(\node')$ have been 
	constructed, $\maplabel(\node')(y) - \maplabel(\node)(x)$ is at least 
	as large as the length of the longest chain between $(\node,x)$ and 
	$(\node',y)$. For the base case $\node=\epsilon$, let 
	$\maplabel(\node)$ be the trivial mapping on the empty domain.
	
	For the induction step, consider a node $\node$. Let $(\node,x_0), 
	(\node,x_1), \ldots$ be the node variables from $\node$ and let 
	$(\node_1,y_1), (\node_2,y_2), \ldots $ be the node variables from 
	all the ancestors of $\node$. Arrange them in ascending order 
	according to $\ro_\tree$. In this arrangement, suppose 
	$(\node_i,y_i) 
	(\node,x_j) (\node,x_{j+1}) \cdots (\node,x_{l}) (\node_{i+1}, 
	y_{i+1})$ is a contiguous sequence of node variables from $\node$ 
	surrounded by ancestor node variables $(\node_i,y_i)$ and 
	$(\node_{i+1}, y_{i+1})$. Set 
	$\maplabel(\node)(x_j)$ to be the sum of $\maplabel(\node_i)(y_i)$ 
	and the length of the longest chain between $(\node_i,y_i)$ and 
	$(\node,x_j)$. Set 
	$\maplabel(\node)(x_{j+1})$ to be the sum of 
	$\maplabel(\node)(x_j)$ and the length of the longest chain between 
	$(\node,x_j)$ and $(\node,x_{j+1})$. Continue this way till 
	$(\node,x_l)$. The value set for $\maplabel(\node)(x_l)$ will be less 
	than $\maplabel(\node_{i+1})(y_{i+1})$ minus the length of the 
	longest chain between $\maplabel(\node)(x_l)$ and 
	$(\node_{i+1},y_{i+1})$, since by induction 
	hypothesis, $\maplabel(\node_{i+1})(y_{i+1}) - 
	\maplabel(\node_{i})(y_{i})$ is large enough to accommodate the 
	longest chain between $(\node_{i},y_{i})$ and 
	$(\node_{i+1},y_{i+1})$ (note that any chain between $(\node,x_j)$ 
	and 
	$(\node,x_{j+1})$ can be concatenated with any chain between 
	$(\node,x_{j+1})$ and $(\node,x_{j+2})$ and so on to form a chain 
	between $(\node_{i},y_{i})$ and $(\node_{i+1},y_{i+1})$). This way, 
	all 
	contiguous sequence of node variables from $\node$ can be 
	mapped satisfactorily. This completes the induction step and hence 
	the proof.
\end{proof}

The above lemma characterizes those strategy trees that are winning strategy trees.
This is the main technical difference between CLTL games and games with register automata
specifications \cite{RFE21,EFK20}. Since register automata can compare values that are
arbitrarily far apart, the corresponding characterization of symbolic structures that
have associated concrete structures is more involved compared to Lemma~\ref{lem:bddChainAdmitsValuation}
above.

Detecting unbounded chains is still difficult for tree automata---to 
find longer chains, we may have to examine longer paths. This difficulty 
can be overcome if we can show that longer chains can be obtained by 
repeatedly joining shorter ones. We now introduce some notation and 
results to formalize this. For a node $\node$ and an ancestor 
$\node_a$, $\hat{\tree}(\node_a,\node)$ is the sequence of frames 
$\tree(\node_a)\cdots \tree(\node)$ labeling the path from $\node_a$ 
to $\node$. A node $\node_1$ is said to occur within the influence of 
$(\node_a,\node)$ if $\node_1$ occurs between $\node_a$ and 
$\node$ or $\node_1$ is an ancestor of $\node_a$ and $|\node_a| - 
|\node_1| \le s-1$, where $s$ is the size of the frame 
$\tree(\node_a)$. The following result follows directly from definitions.
\begin{proposition}[Identical paths induce identical orders]
	\label{prop:identPathIdentOrder}
	Suppose nodes $\node,\node'$ and their ancestors 
	$\node_a,\node_a'$ respectively are such that 
	$\hat{\tree}(\node_a,\node) = \hat{\tree}(\node_a',\node')$. 
	Suppose $\node_1,\node_2$ occur within the influence of 
	$(\node_a,\node)$ and $\node_1',\node_2'$ occur within the 
	influence of $(\node_a',\node')$ such that 
	$|\node|-|\node_1|=|\node'|-|\node_1'|$ and 
	$|\node|-|\node_2|=|\node'|-|\node_2'|$. For any look-ahead 
	variables $x,y$, $(\node_1,x) \ro_\tree (\node_2,y)$ 
	(resp.~$(\node_1,x) \lo_\tree (\node_2,y)$) iff 
	$(\node_1',x) \ro_\tree (\node_2',y)$ (resp.~$(\node_1',x) \lo_\tree 
	(\node_2',y)$).
\end{proposition}

For a node $\node$, the subtree $\tree_\node$ rooted at $\node$ is 
such that for all $\node'$, $\tree_\node(\node') = \tree(\node\cdot 
\node')$. A tree $\tree$ is called \emph{regular} if the set 
$\set{\tree_\node \mid \node \in {\gf}^*}$ is finite, i.e., there are only 
finitely many subtrees up to isomorphism. Two nodes $\node,\node'$ 
are said to be isomorphic if $\tree_\node = \tree_{\node'}$.
\begin{lemma}[Pumping chains in regular trees]
	\label{lem:pumpingRegTree}
	Suppose $\tree$ is a regular tree. Then $\tree$ has unbounded 
	chains iff there exists an infinite path containing two 
	infinite sequences $(\node_1,x),
	(\node_2,x),(\node_3,x) \ldots$ (resp.~$(\node_1',y),
	(\node_2',y),(\node_3',y) \ldots$) such that $\node_{i+1}$ 
	(resp.~$\node_{i+1}'$) is a 
	descendant of $\node_i$ (resp.~$\node_i'$) for all $i \ge 1$ and 
	satisfy one of the following conditions.
	\begin{center}
		$\begin{matrix}
			(\node_1,x) & \rto_\tree &  (\node_2,x) & \rto_\tree & 
			(\node_3,x) \rto_\tree \cdots\\
			\rotatebox[origin=c]{-90}{\ensuremath{\lo_\tree}} & 
			&\rotatebox[origin=c]{-90}{\ensuremath{\lo_\tree}} & &
			\rotatebox[origin=c]{-90}{\ensuremath{\lo_\tree}}\hphantom{abcdab}\\
			(\node_1',y) & \go_\tree &  (\node_2',y) & \go_\tree & 
			(\node_3',y) \go_\tree \cdots
		\end{matrix}$
	    or \enspace
	   	$\begin{matrix}
	    (\node_1,x) & \gto_\tree &  (\node_2,x) & \gto_\tree & 
	    (\node_3,x) \gto_\tree \cdots\\
	    \rotatebox[origin=c]{90}{\ensuremath{\lo_\tree}} & 
	    &\rotatebox[origin=c]{90}{\ensuremath{\lo_\tree}} & 
	    &\rotatebox[origin=c]{90}{\ensuremath{\lo_\tree}}\hphantom{abcdab}\\
	    (\node_1',y) & \ro_\tree &  (\node_2',y) & \ro_\tree & 
	    (\node_3',y) \ro_\tree \cdots
	    \end{matrix}$
	\end{center}
\end{lemma}
\begin{proof}[Proof idea]
	We can choose a chain that is long enough to contain two isomorphic 
	nodes. The path between them can be repeated infinitely. 
	Proposition~\ref{prop:identPathIdentOrder} will imply that this 
	infinite path contains an infinite chain as required.
\end{proof}

\begin{proof}
($\Leftarrow$) We consider the first case; the other case is similar. 
	Since $(\node_i,x) \lo_\tree (\node_i',y) \ro_\tree 
	(\node_{i-1}',y) \ro_\tree \cdots \ro_\tree (\node_1',y)$ for all $i\ge 
	1$, we have $(\node_i,x) \ro_\tree (\node_1',y)$. Hence, 
	$(\node_1,x) \rto_\tree (\node_2,x) \rto_\tree \cdots \rto_\tree 
	(\node_i,x) 
	\ro_\tree (\node_1',y)$ for all $i \ge 1$, demonstrating that there 
	are chains of unbounded lengths between $(\node_1,x)$ and 
	$(\node_1',y)$.
	
	($\Rightarrow$)
	We show the existence of a short segment that can be repeated 
	arbitrarily many times to get the required infinite path. We show that 
	there are node variables along a path satisfying the following 
	conditions:
	\begin{enumerate}
		\item $\begin{matrix}
		(\node_1,x) & \rto_\tree &  (\node_2,x) \\
		\rotatebox[origin=c]{-90}{\ensuremath{\lo_\tree}} & 
		&\rotatebox[origin=c]{-90}{\ensuremath{\lo_\tree}} & \\
		(\node_1',y) & \go_\tree &  (\node_2',y) 
		\end{matrix}$
		or \enspace
		$\begin{matrix}
		(\node_1,x) & \gto_\tree &  (\node_2,x) \\
		\rotatebox[origin=c]{90}{\ensuremath{\lo_\tree}} & 
		&\rotatebox[origin=c]{90}{\ensuremath{\lo_\tree}} & \\
		(\node_1',y) & \ro_\tree &  (\node_2',y) 
		\end{matrix}$,
		\item the nodes are arranged as 
		$\node_1',\node_1, \node_2',\node_2$ in ascending order of 
		depth, $|\node_1'| > k$,
		\item $\node_1,\node_2$ are isomorphic, $\node_1',\node_2'$ 
		are isomorphic and 
		$|\node_1|-|\node_1'|=|\node_2|-|\node_2'|\le k$.
	\end{enumerate}
	The node variables mentioned above are as shown below.
	\begin{center}
		\begin{tikzpicture}
		\node[state, label=180:root] (root) at (0cm,0cm) {};
		
		\node[state, label=-90:{$(\node_1',y)$}] (n1p) at 
		([xshift=2cm]root) {};
		\node[state, label=90:{$(\node_1,x)$}] (n1) at 
		([xshift=0.7cm]n1p) {};
		
		\node[state, label=-90:{$(\node_2',y)$}] (n2p) at 
		([xshift=2.3cm]n1) {};
		\node[state, label=90:{$(\node_2,x)$}] (n2) at 
		([xshift=0.7cm]n2p) {};
		
		\node[state, label=-90:{$(\node_3',y)$}] (n3p) at 
		([xshift=2.3cm]n2) {};
		\node[state, label=90:{$(\node_3,x)$}] (n3) at 
		([xshift=0.7cm]n3p) 
		{};
		
		\draw[-] (root) -- (n3);
		\draw[decorate, decoration=brace] ([yshift=0.7cm]n1.center) 
		-- node[auto=left] {pattern} 
		([yshift=0.7cm,xshift=-0.1cm]n2.center);
		\draw[decorate, decoration=brace] ([yshift=0.7cm, 
		xshift=0.1cm]n2.center) --  node[auto=left] {pattern repeats} 
		([yshift=0.7cm]n3.center);
		\end{tikzpicture}
	\end{center}
	We first prove that the existence of such nodes is sufficient. Since 
	$\node_1,\node_2$ are isomorphic, for any sequence of frames 
	starting from $\node_1$, the same sequence also starts from 
	$\node_2$. Hence there is a descendant $\node_3$ of $\node_2$ 
	such 
	that $\node_2,\node_3$ are isomorphic and 
	$\hat{\tree}(\node_1,\node_2)=\hat{\tree}(\node_2,\node_3)$. The 
	nodes $\node_1', \node_1, \node_2', \node_2$ occur within the 
	influence of $(\node_1,\node_2)$ and the nodes 
	$\node_2',\node_2, 
	\node_3', \node_3$ occur within the influence of 
	$(\node_2,\node_3)$. 
	In the first case in the first condition above, $(\node_1,x) \rto_\tree 
	(\node_2,x) \lo_\tree 
	(\node_2',y) 
	\ro_\tree (\node_1',y)$ and 
	Proposition~\ref{prop:identPathIdentOrder} 
	implies that $(\node_2,x) \rto_\tree (\node_3,x) \lo_\tree 
	(\node_3',y) \ro_\tree (\node_2',y)$. This pattern can be repeated 
	arbitrarily many times, proving that there are node variables as 
	stated in the first case of the lemma. The other case is similar.
	
	Now we will show the existence of the short segment as claimed 
	above. 
	Since $\tree$ is regular, the number of non-isomorphic subtrees of 
	$\tree$ is finite, say $\kappa$. Let $N=\kappa^2 |\lav|^2$.
	We will show subsequently that there is a chain of the form 
	$(\node,x_1) \rto_\tree (\node_1,y_1) \rto_\tree (\node_2,y_2) 
	\rto_\tree \cdots \rto_\tree (\node_{N+2},y_{N+2}) \ro_\tree 
	(\node',x_2)$ or $(\node,x_1) \gto_\tree (\node_1,y_1) \gto_\tree 
	(\node_2,y_2) 
	\gto_\tree \cdots \gto_\tree (\node_{N+2},y_{N+2}) \go_\tree 
	(\node',x_2)$, 
	where $\node_1$ is a descendant of both $\node$ and $\node'$ of 
	depth at least $(k+1)$ more than both $\node$ and $\node'$ and 
	$\node_{i+1}$ is a descendant of $\node_i$ of depth at least 
	$(k+1)$ 
	more than $\node_i$ for all $i \in [1,N+1]$ (we call such chains 
	straight 
	segments). We will only consider the 
	first case here; the other case is similar. Now 
	$(\node_{N+2},y_{N+2}) 
	\ro_\tree (\node',x_2)$ and $\node_{N+2}$ is  a deep descendant of 
	$\node'$ with $\node_1, \ldots, \node_{N+1}$ (which are 
	themselves 
	at least $(k+1)$ positions apart from each other) in between. Recall 
	that 
	$\ro_\tree$ is the transitive closure of $\lo_\tree$ and $\lo_\tree$ 
	holds only between node variables that are at most $k$  positions 
	apart. Hence, there must be intermediate node variables between 
	$(\node_{N+2},y_{N+2}) , (\node',x_2)$ so that 
	$(\node_{N+2},y_{N+2}) 
	\ro_\tree (\node',x_2)$. For every $i \in [1,N+1]$, there must be 
	some 
	intermediate node variable $(\node_i',y_i')$ such that $\node_i'$ is 
	an 
	ancestor of $\node_i$, $|\node_i|-|\node_i'|\le k$ and 
	$(\node_{N+2},y_{N+2}) \ro_\tree 
	(\node_i',y_i') \ro_\tree 
	(\node',x_2)$. Since $|\node_i|-|\node_i'|\le k$, either 
	$(\node_i,y_i) 
	\lo_\tree (\node_i',y_i')$ or $(\node_i',y_i') \lo_\tree (\node_i,y_i)$ 
	(the frame $\tree(\node_i)$ spans $\node_i'$ also; hence the 
	frame imposes an order between the node variables). If 
	$(\node_i',y_i') 
	\lo_\tree (\node_i,y_i)$, then $(\node_i,y_i) \rto_\tree 
	(\node_{N+2},y_{N+2}) \ro_\tree (\node_i',y_i') \lo_\tree 
	(\node_i,y_i)$ 
	implies that $(\node_i,y_i) \rto_\tree (\node_i,y_i)$, contradicting the 
	fact that $\rto_\tree$ is irreflexive. Hence, $(\node_i,y_i) 
	\lo_\tree (\node_i',y_i')$. Consider the sequence $(\node_1,y_1), 
	(\node_1',y_1'), (\node_2,y_2), (\node_2',y_2'), \ldots, 
	(\node_{N+1},y_{N+1}), (\node_{N+1}',y_{N+1}')$. Since 
	$N=\kappa^2 
	|\lav|^2$, there are $i,j$ such that $\node_i$ (resp.~$\node_i'$) is 
	isomorphic to $\node_j$ (resp.~$\node_j'$), $y_i=y_j$ and 
	$y_i'=y_j'$. 
	The node variables $(\node_i,y_i), (\node_j,y_i), (\node_i',y_i'), 
	(\node_j',y_i')$ satisfy the conditions required for $(\node_1,x), 
	(\node_2,x), (\node_1',y), (\node_2',y)$ respectively in our claim 
	about 
	the existence of a short segment.
	
	Next we will show that there are chains that go arbitrarily deep in a 
	single branch. Suppose there are chains of unbounded lengths 
	between 
	$(\node_1,x_1)$ and $(\node_2,x_2)$. All such chains must pass 
	through the least common ancestor (say $\node_a$) of 
	$\node_1,\node_2$. For some variable $x_a$, there must be chains 
	of 
	unbounded lengths between either $(\node_1,x_1)$ and 
	$(\node_a,x_a)$ or between $(\node_a,x_a)$ and $(\node_2,x_2)$. 
	Say 
	there are unbounded chains between $(\node_1,x_1)$ and 
	$(\node_a,x_a)$; the other case is similar. There is only one path 
	between $\node_1$ and $\node_a$, so there must be chains of 
	unbounded lengths that go beyond this path and come back. There 
	must be node variables $(\node_1,y_1), (\node_1,y_2)$ or 
	$(\node_a,y_1), (\node_a,y_2)$ such that there are chains of 
	unbounded lengths between them. We will consider $(\node_1,y_1), 
	(\node_1,y_2)$; the other case is similar. For the chains of 
	unbounded 
	lengths starting from $(\node_1,y_1)$ and ending at 
	$(\node_2,y_2)$, 
	let $\node$ be the highest node (nearest to the root) visited. There 
	must be $(\node,z_1), (\node,z_2)$ such that there are chains of 
	unbounded lengths between them that only visit descendants of 
	$\node$. If there 
	is a bound (say $B$) on how deep the chains go below $\node$ and 
	come back, 
	the number of nodes that can be visited is bounded by the number 
	of 
	node variables that occur in the subtree of height $B$ rooted at 
	$\node$ (a node can occur 
	at most once in a chain; otherwise, it will contradict the fact that $
	\rto_\tree$ is irreflexive). Hence, for any bound $B$, there are chains 
	that go deeper than $B$ and come back.
	
	Next we prove that there is no bound on the number of node 
	variables 
	in a single path that belong to a chain. For this, first suppose that 
	there 
	is a 
	node $\node$ and a chain goes down one child of $\node$ starting 
	from $(\node,x)$, comes back to $\node$ via $(\node,y)$ and goes 
	down another child. Then we have $(\node,x) \rto_\tree (\node,y)$ 
	or 
	$(\node,y) \rto_\tree (\node,x)$ (see the illustration below; if 
	$(\node,x) \ro_\tree (\node_b,x') \rto_\tree (\node_b,y') \ro_\tree 
	(\node,y)$ in the branch, we have $(\node,x) \rto_\tree (\node,y)$ in 
	the main path by transitivity). Hence, every such node contributes a 
	node variable in a 
	chain.
	\begin{center}
		\begin{tikzpicture}
		\node[state, label=90:root](r) at (0cm,0cm) {};
		\node[state, label=155:{$(\node,x)$}](n1) at ([xshift=2cm]r) {};
		\node[state, label=25:{$(\node,y)$}](n4) at ([xshift=0.75cm]n1) {};
		\node[state, label=-135:{$(\node_b,x')$}](n2) at 
		([xshift=1cm,yshift=-1cm]n1) {};
		\node[state, label=-45:{$(\node_b,y')$}](n3) at 
		([xshift=1cm,yshift=-1cm]n4) {};
		
		\draw[thick] (root) -- (n1);
		\draw[very thin] (n1) -- (n2) -- node[auto=right] {$\rto_\tree$} 
		(n3) -- (n4);
		\draw[thick] (n1) -- node[auto=left] (l1) {$\rto_\tree$} (n4);
		\node[draw=black, fit=(n1) (n4) (l1)] (b1) {};
		
		\node[state](n5) at ([xshift=2cm]n4) {};
		\node[state](n8) at ([xshift=0.75cm]n5) {};
		\node[state](n6) at ([xshift=1cm,yshift=-1cm]n5) {};
		\node[state](n7) at ([xshift=1cm,yshift=-1cm]n8) {};
		
		\draw[thick] (n4) -- (n5);
		\draw[very thin] (n5) -- (n6) -- node[auto=right] {$\rto_\tree$} 
		(n7) -- (n8);
		\draw[thick] (n5) -- node[auto=left] (l2) {$\rto_\tree$} (n8);
		\node[draw=black, fit=(n5) (n8) (l2)] (b2) {};
		
		\draw[thick] (n8) -- ([xshift=1cm]n8.center);
		\draw[dotted] ([xshift=1cm]n8.center) -- ([xshift=2cm]n8.center);
		\node (b4) at ([xshift=3cm]n8.center) {main path};
		\node (b5) at ([yshift=-1cm]b4.center) {branches};
		
		\node at ([yshift=1cm]barycentric cs:b1=1,b2=1) (b3) {branching 
			nodes};
		\draw[->] (b3) -- (b1);
		\draw[->] (b3) -- (b2);
		\end{tikzpicture}
	\end{center}
	So if there is no bound on the number of such branching nodes along 
	a 
	path, then there is no bound on the number of node variables 
	in a single path that belong to a chain, 
	as required. Suppose for the sake of contradiction that the 
	number of such branching nodes along any path is bounded (by say 
	$B_1$) and the number of node variables in a chain along any one 
	path is also bounded (say by $B_2$). Then any chain is in a subtree 
	with at most $|\gf|^{B_1}$ leaves (and hence at most as many paths) 
	and at most $B_2$ node variables along any path, so the length of 
	such chains is bounded. Hence, either the number of branching 
	nodes 
	along a path is unbounded or the number of node variables in a 
	chain 
	along a path is unbounded. Both of these imply that the number of 
	node variables in a chain along a path is unbounded, as required.
	
	A chain that goes deep down a path may make u-turns (first descend 
	through descendants and then go to an ascendant or vice-versa) 
	multiple 
	times within the branch. We would like to prove that there is no 
	bound 
	on the length of chain segments that don't have u-turns (these are 
	the 
	straight segments that we need). Suppose for the sake of 
	contradiction 
	that there is a 
	bound on the length of straight segments. Then there is no bound 
	on the number of straight segments in a path, since we have already 
	shown that the number of 
	node variables in a chain along a path is unbounded. There can be 
	only 
	boundedly 
	many distinct straight segments in a path of bounded depth, so the 
	straight segments go deeper without any bound. If there is a straight 
	segment and another one occurs below the first one, the first straight 
	segment can be extended by appending node variables of the 
	second 
	one, as can be seen in the illustration below.
	\begin{center}
		\begin{tikzpicture}
		\node[state, label=180:root] (r) at (0cm,0cm) {};
		\node[state] (n1) at ([xshift=2cm]r) {};
		\node[state] (n2) at ([xshift=1cm]n1) {};
		\node[state] (n3) at ([xshift=1cm]n2) {};
		\node[state, label=90: first straight segment] (n4) at 
		([xshift=1cm]n3) {};
		
		\node[coordinate] (n5) at ([yshift=-0.6cm]n4) {};
		\node[state] (n6) at ([xshift=-0.5cm]n5) {};
		\node[state] (n7) at ([xshift=-1cm]n6) {};
		\node[coordinate] (n8) at ([yshift=-0.6cm]n7) {};
		\node[state, label=-90:second straight segment] (n9) at 
		([xshift=0.5cm]n8) {};
		\node[state] (n10) at ([xshift=1cm]n9) {};
		\node[state] (n11) at ([xshift=1cm]n10) {};
		\node[state] (n12) at ([xshift=1cm]n11) {};
		
		\draw[thick, rounded corners=0.05cm] (r) -- (n1) -- (n2) -- (n3) 
		-- (n4) -- (n5) -- (n6) --	(n7) -- (n8) -- (n9) -- (n10) -- (n11) -- 
		(n12);
		
		\draw[very thin] (n4) -- node[auto=left, pos=0.6] {first segment 
			extended} 
		(n11);
		\end{tikzpicture}
	\end{center}
	This contradicts the hypothesis that length of straight segments is 
	bounded. This shows that there are unboundedly long straight 
	segments, completing the proof.
\end{proof}

%% file: treeAutomaton.tex
Lemma~\ref{lem:pumpingRegTree} says that if a regular tree has unbounded chains, it will have 
an infinite path containing an infinite chain. The infinite sequence of the 
first (resp.~second) kind given in 
Lemma~\ref{lem:pumpingRegTree} is called an infinite forward 
(resp.~backward) chain. Now we design a tree 
automaton $\mathcal{A}_{\phi}$ whose language 
$\mathcal{L(\mathcal{A}_{\phi})}$ is an 
approximation 
of the set $\mathcal{T} = \set{\tree \mid \exists \maplabel, (\tree,\maplabel) \text{ 
is a winning strategy tree}}$ such that $\mathcal{L(\mathcal{A}_{\phi})}$ 
is non-empty iff 
$\mathcal{T}$ is. Hence, the single-sided CLTL realizability problem is 
equivalent to checking the non-emptiness of 
$\mathcal{L(\mathcal{A}_{\phi})}$. The tree automaton 
$\mathcal{A}_{\phi}$ is defined as the intersection of three automata 
${\mathcal{A}}^{\text{str}}_{\phi}$, ${\mathcal{A}}^{\text{symb}}_{\phi}$ 
and 
${\mathcal{A}}^{\text{chain}}_{\phi}$, all of which read $|\gf|$-ary trees 
labeled with letters from $\frames$. 
The automaton ${\mathcal{A}}^{\text{str}}_{\phi}$ 
accepts the set of all strategy trees, 
${\mathcal{A}}^{\text{symb}}_{\phi}$ accepts the set of all trees each of 
whose paths symbolically satisfies the formula $\phi$ 
and ${\mathcal{A}}^{\text{chain}}_{\phi}$ accepts 
the set of all trees that do not have any infinite forward or backward 
chains. We now give a detailed construction of these automata.

The automaton ${\mathcal{A}}^{\text{str}}_{\phi}$ has set of states 
$\frames$. In state $f$, it can read the input label $f$ and go to states 
$f_1, \ldots, f_{|\gf|}$ in its children, provided $(f,f_i)$ is one-step 
compatible and $({\gap}_i, f_i)$ is gap-compatible for all $i \in 
[1,|\gf|]$. All states are accepting in this B\"{u}chi automaton.
This automaton just checks that every pair of consecutive frames along 
every branch of the tree is one-step compatible and gap-compatible 
and hence verifies that the tree accepted is a strategy tree. Now, the size of the set of states of ${\mathcal{A}}^{\text{str}}_{\phi}$ is $|\frames|$, and the size of the transition set is $|\frames| \times |\Sigma| \times {|\frames|}^{|\gf|}$ where the input alphabet $\Sigma = \frames$. Since, $\gf$ is the set of all gap functions associated with mappings of the form $\efbv \to \Int$, by definition of $\gf$ its range must be $\set{0,\ldots,|\efbv|^2}$ implying $|\gf| \le |\efbv|^{({|\efbv|}^2)}$. Also, from the definition of $\frames$, we get $|\frames| \le 2^{{(k.|\lav|)}^2} \times (|\fbv|^{|\fbv|^2})^k$ where $k$ is the $\nxt$-length of $\phi$. Thus, the size of ${\mathcal{A}}^{\text{str}}_{\phi}$ is double exponential in the size of $\phi$.

The automaton ${\mathcal{A}}^{\text{symb}}_{\phi}$ checks that every 
path in the input tree is accepted by a B{\"u}chi automaton 
${\mathcal{B}}^{\text{ symb}}_{\phi}$, which ensures that the input 
sequence symbolically satisfies the formula $\phi$. Given the B{\"u}chi 
automaton ${\mathcal{B}}^{\text{ symb}}_{\phi}$, we first convert it to some deterministic parity automaton ${\mathcal{C}}^{\text{ symb}}_{\phi}$ in exponential time in the size of ${\mathcal{B}}^{\text{ symb}}_{\phi}$ and from that, it is easy to construct 
the parity tree automaton ${\mathcal{A}}^{\text{symb}}_{\phi}$ with the same size as ${\mathcal{C}}^{\text{ symb}}_{\phi}$.
The B{\"u}chi automaton ${\mathcal{B}}^{\text{ symb}}_{\phi}$ needs to 
check symbolic satisfiability---whether an atomic formula is satisfied 
at a position can be decided by checking just the current frame, just 
like propositional LTL. Hence the standard B{\"u}chi automaton 
construction for LTL can be used to construct 
${\mathcal{B}}^{\text{ symb}}_{\phi}$ in \textsc{EXPTIME} \cite{VW86}. Thus, the parity tree automaton ${\mathcal{A}}^{\text{symb}}_{\phi}$ can be constructed in \textsc{2EXPTIME} in the size of $\phi$.

Next, we describe the construction of the parity tree automaton
${\mathcal{A}}^{\text{chain}}_{\phi}$. It needs to check that there are no 
infinite forward or backward chains in any of the paths. For this we will first construct a B\"{u}chi word automaton that accepts all words not having an infinite forward or backward chain, convert it into a deterministic parity automaton ${\mathcal{C}}^{\text{chain}}_{\phi}$ and then as before, construct ${\mathcal{A}}^{\text{chain}}_{\phi}$ with the same size as ${\mathcal{C}}^{\text{chain}}_{\phi}$.
This B\"{u}chi word automaton can be constructed by complementing the B\"{u}chi automaton ${\mathcal{B}}^{\text{chain}}$ which accepts all words that contain an infinite forward chain or 
an 
infinite backward chain in \textsc{EXPTIME} in the size of ${\mathcal{B}}^{\text{chain}}$ \cite{SVW87}. The construction of such a B\"{u}chi automaton ${\mathcal{B}}^{\text{chain}}$ is already described in 
\cite{DD07}, which we reproduce here using our notation. Recall that 
$\ro_\tree$ is the transitive closure of $\lo_\tree$. So the part of the 
infinite chain $(\node_1,x) \rto_\tree (\node_2,x) \rto_\tree 
(\node_3,x) \rto_\tree \cdots$ may be embedded in a sequence  
$(\node_1,x) \lo_\tree (\node_1'',z_1) \rto_\tree (\node_2,x) \lo_\tree 
(\node_2'',z_2) \rto_\tree (\node_3,x) \lo_\tree \cdots$. So the 
B\"{u}chi word automaton checks for a sequence of node variables 
related by $\lo_\tree$, with the order being strict infinitely often.

Define ${\mathcal{B}}^{\text{chain}} = (Q, \Sigma, \set{q_0}, 
\longrightarrow, F)$, where:

\begin{itemize}
	\item $Q = \set{q_0} \cup (\lav \times \set{0,...,(k-1)} \times \lav 
	\times 
	\set{0,...,(k-1)} \times \set{d,e} \times \set{0,1})$; (Here $d$ and 
	$e$ 
	denote the forward and backward infinite sequences and $1$ or $0$ 
	indicate whether or not the order is strict)
	
	\item $\Sigma = \frames$
	
	\item $\longrightarrow$ is given by:
	\begin{itemize}
		\item $q_0 \overset{f}{\longrightarrow} q_0$
		\item $q_0 \overset{f}{\longrightarrow} (x,i,y,j,d,0)$ and $q_0 
		\overset{f}{\longrightarrow} (x,i,y,j,e,0)$       $\forall x,i,y,j$\\
		Here $(x,i)$ and $(y,j)$ indicate the variables $x$ and the $y$ at 
		the $i^{\text{th}}$ and $j^{\text{th}}$ positions of the current 
		frame. The automaton guesses $(x,i)$ as $(\node_1, x)$ and 
		$(y,j)$ 
		as $({\node}'_1, y)$, as well as which of $d$ or $e$ forms a 
		chain.\\
		
		\item $(x,i,y,j,\delta,b) \overset{f}{\longrightarrow} (x, i-1, y, j-1, 
		\delta, b)$ for $\delta \in \set{d,e}$ and $b \in \set{0,1}$ 
		provided 
		$i, j \geq 2$\\
		(Wait till $x$ or $y$ is at the edge of the frame).\\
		\item $(x,1,y,j,\delta,b) \overset{f}{\longrightarrow} (z, i, y, j-1, 
		\delta, b')$ provided \\
		$j > 1$,\\
		$x <_f {\nxt}^i z$ and $b' = 1$, or $x \equiv_f {\nxt}^i z$ and 
		$b' = 0$ 
		and\\
		${\nxt}^i z <_f {\nxt}^{j-1} y$.\\
		(Guess a continuation of the chain from $(x,1)$).\\
		\item $(x,i,y,1,d,b) \overset{f}{\longrightarrow} (x, i-1, w, j, d, 
		b')$ provided \\
		$i > 1$,\\
		${\nxt}^j w\text{ }{\leq}_f\text{ }y$ and $b' = b$ and \\
		${\nxt}^{i-1} x <_f {\nxt}^{j} w$.\\
		(Guess a continuation of the chain from $(y,1)$).
		\item similarly for $(x,1)$ and $(y,1)$ simultaneously.
		\item similar transitions for $e$
	\end{itemize}
	\item the set of accepting states comprises all states of the form 
	$(x,i,y,j,\delta,1)$
\end{itemize}

The size of ${\mathcal{B}}^{\text{chain}}$ is polynomial in the size of the CLTL formula $\phi$ and hence, the size of ${\mathcal{A}}^{\text{chain}}_{\phi}$ is double exponential in the size of $\phi$.

Now we have the following result.

\begin{lemma}
\label{lem:treeautomata}
The \sys player has a winning strategy in the single-sided 
CLTL$(\Int,<,=)$ game with winning condition $\phi$ iff 
$\mathcal{L(\mathcal{A}_{\phi})}$ is non-empty.
\end{lemma}
\begin{proof}
Suppose there is a winning strategy for the \sys player in 
single-sided CLTL$(\Int,<,=)$ game with winning condition 
$\phi$. By Lemma~\ref{lem:stratToTree}, there exists a winning 
strategy tree, say $(\tree,\maplabel)$. 
Since, $\tree$ is a strategy tree, $\tree \in 
\mathcal{L}({\mathcal{A}}^{\text{str}})$. We know that every branch of 
$\tree$ must symbolically satisfy $\phi$ and hence, $\tree \in 
\mathcal{L}({\mathcal{A}}^{\text{symb}}_{\phi})$. Further, since $\tree$ has 
the labelling function $\maplabel$, Lemma~\ref{lem:bddChainAdmitsValuation} implies that $\tree$ has bounded 
chains and thus, it cannot have any infinite forward or backward chains. 
So $\tree \in \mathcal{L}({\mathcal{A}}^{\text{chain}})$. Thus, $\tree \in 
\mathcal{L(\mathcal{A}_{\phi})}$.

Conversely, suppose $\mathcal{A}_{\phi}$ accepts a tree $\tree$. Then, 
using \cite[Lemma 8.16, Theorem 8.19, Corollary 8.20]{N02} (suitably 
modifying the proofs to work for tree automata that accept $|\gf|$-ary 
trees), we can see that $\mathcal{A}_{\phi}$ must accept a 
regular tree $\tree'$. Since, $\tree' \in 
\mathcal{L(\mathcal{A}_{\phi})}$, every branch of $\tree'$ 
must symbolically satisfy $\phi$, $\tree'$ must be a strategy tree and it 
cannot have any infinite forward or backward chains. Thus, by Lemma~\ref{lem:pumpingRegTree}, $\tree'$ must 
have bounded chains and hence by Lemma~\ref{lem:bddChainAdmitsValuation}, $\tree'$ must have a 
labelling 
function $\maplabel'$ such that $(\tree',\maplabel')$ is a winning strategy tree. Hence, by 
Lemma~\ref{lem:stratToTree} the 
\sys player has a winning strategy in the single-sided 
CLTL$(\Int,<,=)$ game.
\end{proof}

\begin{theorem}
\label{thm:decresult}
The single-sided realizability problem for CLTL over $(\Int,<,=)$ is \textsc{2EXPTIME}-complete.
\end{theorem}
\begin{proof}
Given a formula $\phi$, Lemma~\ref{lem:treeautomata} 
implies that it is enough to construct 
the tree automaton $\mathcal{A}_{\phi}$ and check it for 
non-emptiness. From the description of the construction in Appendix C, we can see that ${\mathcal{A}}^{\text{str}}_{\phi}$, ${\mathcal{A}}^{\text{symb}}_{\phi}$ and ${\mathcal{A}}^{\text{chain}}_{\phi}$ can be constructed in \textsc{2EXPTIME} in the size of $\phi$. Thus, the automaton $\mathcal{A}_{\phi}$ can be constructed in \textsc{2EXPTIME}. Now, checking non-emptiness of a parity tree automaton is decidable and the upper bound stated in \cite[Corollary 8.22 (1)]{N02} implies that the single-sided realizability problem for CLTL over $(\Int,<,=)$ is in \textsc{2EXPTIME}. 
Now, the realizability problem for LTL is \textsc{2EXPTIME}-complete \cite{PR89} and hence, the single-sided realizability problem for CLTL over $(\Int,<,=)$ must also be \textsc{2EXPTIME}-complete. 
\end{proof}

%% file: promptCLTL.tex
Prompt-LTL is an extension of LTL with the prompt-eventually operator 
$\prompt$. The realizability problem for prompt-LTL is decidable, as 
shown in \cite{KPV09} via a reduction to LTL 
realizability problem. We consider a similar extension of CLTL. We show 
that the single-sided realizability problem for prompt-CLTL over $(\Int,<,=)$ is 
decidable, by giving a reduction to single-sided 
games on CLTL. We adapt techniques from \cite{KPV09} to show 
this.

The syntax of prompt-CLTL is given by the grammar
$\phi ::= c ~|~ \lnot \phi ~|~ \phi \lor \phi ~|~ \nxt \phi ~|~ \prompt 
~ \phi ~|~\phi \until \phi$.
The semantics of prompt-CLTL are defined over concrete models 
$\sigma$ and a bound $k \ge 0$. We say $(\sigma, i, k) \models \phi$ 
to indicate that $\phi$ holds in position $i$ of $\sigma$ with bound 
$k$. The relation $\models$ is defined by induction on the structure of 
$\phi$ similar to CLTL, except for $\prompt ~ \phi$ that is defined as 
follows: $(\sigma, i, k) \models \prompt ~ \phi$ iff there exists $j$ 
such that $i \le j \le i+k$ and $(\sigma, j, k) \models \phi$.
We say $(\sigma, k) \models \phi$ if $(\sigma, 0, k) \models \phi$.
Single-sided prompt-CLTL games are similar to single-sided CLTL 
games, except that the winning condition is a 
prompt-CLTL formula.

A strategy $\sstrat \colon M^* \cdot\eMap \to \sMap$ is said to be winning 
for \sys if there is a bound $k \ge 0$ such that for all 
models $\sigma$ generated from plays conforming to $\sstrat$, 
$(\sigma, k) \models \phi$. Given, a prompt-CLTL formula $\phi$, we define the single-sided realizability problem for prompt-CLTL as the problem of checking whether $\sys$ has a winning strategy in the single-sided prompt-CLTL game with winning condition $\phi$.

We now describe the alternating colour technique, as proposed in 
\cite{KPV09}, with minor modifications required to lift it to 
prompt-CLTL. Let $\splx$ and $\sply$ be new \sys look-ahead 
variables not in $V$ and let $\col$ be the constraint $\splx = \sply$. A $\col$-colouring of a concrete model $\sigma$ 
is a concrete model $\sigma'$ over the set of
variables $V \cup \set{\splx, \sply}$, such that $\sigma'$ agrees with 
$\sigma$ on $V$. 
We say that position $i$ is green (resp.~red) if $\col$ is true 
(resp.~false) in ${\sigma}_i'$. A position 
$i$ is said to be a $\col$-change point if either $i=0$ or the colours at 
$i-1$ and $i$ are different. A subword ${\sigma}_i'...{\sigma}_{i'}'$ is a 
$\col$-block if all positions in the subword have the same colour, and 
$i$ and $i'+1$ are $\col$-change points. For $k \ge 0$, we say that 
$\sigma'$ is 
$k$-spaced/$k$-bounded/$k$-tight (with respect to the colouring) if 
$\sigma'$ has infinitely many blocks, and all the blocks are of length at 
least $k$, at most $k$ and exactly $k$, respectively.

For a 
prompt-CLTL formula $\phi$ over $V$, we define a formula $c(\phi)$ 
over $V \cup\set{\splx, \sply}$ as $
c(\phi) := (GF \col \wedge GF \neg \col) \wedge {rel}_{\col}(\phi)$,
where ${rel}_{\col}(\phi)$ denotes the formula obtained from $\phi$ by 
recursively, replacing every subformula of the form $\prompt ~ \psi$ 
by the CLTL formula $(\col \implies (\col \until (\neg \col \until \psi))) 
\wedge (\neg \col \implies (\neg \col \until (\col \until \psi)))$. The formula $c(\phi)$ forces every concrete model to be partitioned into infinitely many blocks and requires each prompt eventuality to be satisfied in the current or the next block or the position immediately after the next block. 

We have the following lemma from \cite{KPV09} (stated in the context 
of prompt-CLTL). The proof of the lemma is almost the same as that in 
\cite{KPV09} and hence, we omit the proof over here. 

\begin{lemma}
	\label{lem:altcolour}
	Consider a prompt-CLTL formula $\phi$, a concrete model 
	$\sigma$, and a bound $k \ge 0$, 
	\begin{enumerate}
		\item If $(\sigma,k) \models \phi$, then for every $k$-spaced 
		$\col$-colouring $\sigma'$ of $\sigma$, we have, $\sigma' 
		\models c(\phi)$.
		
		\item If $\sigma'$ is a $k$-bounded $\col$-colouring of 
		$\sigma$ such that $\sigma' \models c(\phi)$, then $(\sigma, 2k) 
		\models \phi$
	\end{enumerate}
\end{lemma}

We now have the following theorem.
\begin{theorem}
\label{thm:promptCLTL_to_CLTL_red}
Over the 
domain $(\Int,<,=)$, \sys has a winning strategy in the single-sided 
prompt-CLTL game with formula $\phi$ over the set of variables $V$  
iff she has a winning strategy in the single-sided CLTL game with
formula $c(\phi)$ over the set of variables $V \cup\set{\splx, \sply}$ .
\end{theorem}

\begin{proof}
	($\Rightarrow$) Suppose $\phi$ is a prompt-CLTL formula which is 
	single-sided realizable. Then there exists a strategy $\sstrat \colon 
	M^*\cdot\eMap \to \sMap$ for the \sys player in the single-sided 
	prompt-CLTL game with winning condition $\phi$, over 
	$(\Int,<,=)$, 
	and a bound $k \ge 0$, such that for all concrete models resulting 
	from 
	a play conforming to $\sstrat$, $(\sigma,k) \models \phi$. We shall 
	now extend the strategy $\sstrat$ to a strategy $\sstrat'$ for the 
	\sys 
	player in the single-sided CLTL game with winning condition 
	$c(\phi)$.
	
	Let $M'$ (resp.~$\sMap'$) denote the set of all mappings of the form 
	$V' \to \Int$ (resp.~$\sv' \to \Int$). 
	We define $\sstrat' \colon M'^* \cdot\eMap \to \sMap'$ as: for all $\tau 
	\in 
	M^* \cdot\eMap$, $\sstrat'(\tau) = \sstrat(\tau) \cup \set{\splx 
		\mapsto 
		0, 
		\sply \mapsto 0}$ if $|\tau|\hspace{0.1cm} mod \hspace{0.1cm} 
		2k$ 
	lies between $0$ and $k-1$ and $\sstrat'(\tau) = \sstrat(\tau) \cup 
	\set{\splx \mapsto 0, \sply \mapsto 1}$ if 
	$|\tau|\hspace{0.1cm}mod\hspace{0.1cm}2k$ lies between $k$ and 
	$2k-1$. Now, any concrete model $\sigma'$ generated by a play in 
	the 
	CLTL game with winning condition $c(\phi)$ that conforms to 
	$\sstrat'$, is $k$-tight (by construction). Thus, by 
	Lemma~\ref{lem:altcolour}, $\sigma' \models c(\phi)$ and hence 
	$c(\phi)$ is realizable.
	
	($\Leftarrow$) Suppose the CLTL formula $c(\phi)$ is realizable by a 
	single-sided CLTL game with winning condition $c(\phi)$ over 
	$(\Int,<,=)$. By Lemma \ref{lem:treeautomata}, the language of the 
	automaton $\mathcal{A}^{\phi}$ (defined in the previous section) is 
	non-empty. Thus, $\mathcal{A}^{\phi}$ accepts a tree $\tree'$. 
	Now, as we saw in the proof of Lemma \ref{lem:treeautomata}, 
	$\mathcal{A}^{\phi}$ must also accept a regular strategy tree $\tree \colon 
	{\gf}^* \to \frames'$ (with respect to the winning condition formula 
	$c(\phi)$), with a labelling function $\maplabel$ such that  
	$(\tree,\maplabel)$ is a winning strategy tree. Thus, the infinite 
	sequence $\tree(\pth)$ along every infinite path $\pth$ of the tree 
	$\tree$, satisfies the formula $c(\phi)$. 
	
	Now, we define a strategy $\sstrat' \colon {\Map}^* \cdot{\eMap} \to \sMap$, for the 
	\sys player in the single-sided prompt-CLTL game with winning 
	condition $\phi$. For all $\tau = (em_1 \oplus sm_1 )...(em_{|\tau|-1} \oplus sm_{|\tau|-1})\cdot em_{|\tau|} \in 
	{\Map}^* \cdot{\eMap}$, let $\gap(\tau) = 
	\gap(em_1)\gap(em_2)...\gap(em_{|\tau|})$. Now, define 
	$\sstrat'(\tau) = \hat{\maplabel}(\gap(\tau)) \restr \sv$. We shall 
	prove that $\sstrat'$ is a winning strategy for the \sys player.
	
	We know that a regular tree has a finite number of non-isomorphic 
	subtrees. Let us say the number of non-isomorphic subtrees of 
	$\tree$ is $\kappa$. We will show that every concrete model 
	$\sigma$ admitted by a symbolic model $\tree(\pth)$ along some 
	path $\pth$ of the strategy tree $\tree$ is $(\kappa + 1)$-bounded.
	
	Assume for the sake of contradiction that $\sigma$ has adjacent 
	$\col$-change points $i$ and $j$ such that $j - i > \kappa + 1$. 
	Let ${\node}_1{\node}_2{\node}_3...$ be an infinite path along 
	$\tree$ such that the infinite sequence of frames along this path 
	admits the concrete model $\sigma$. Since, the number of 
	non-isomorphic subtrees is $\kappa$, there must exist positions 
	$i'$ and $j'$ such that $i \le i' < j' \le j-1$ and the subtrees rooted 
	at the nodes ${\node}_{i'}$ and ${\node}_{j'}$ are isomorphic. This 
	means that there exists an infinite path in $\tree$, such that the 
	infinite sequence of frames along that path equals 
	$\tree({\node}_{1})\tree({\node}_{2})...\tree({\node}_{i'-1}){(\tree({\node}_{i'})...\tree({\node}_{j'-1}))}^{\omega}$.
	 But no concrete model admitted by this sequence of frames can 
	satisfy $GF \col \wedge GF \neg \col$, which contradicts the fact 
	that $\tree$ is a winning strategy tree with respect to the winning 
	condition $c(\phi)$.
	
	Thus, every concrete model admitted by a symbolic model 
	$\tree(\pth)$ along some path $\pth$ of the strategy tree $\tree$ is 
	$(\kappa + 1)$-bounded, and satisfies $c(\phi)$. Therefore, by 
	Lemma \ref{lem:altcolour}, every concrete model $\sigma'$ 
	generated during a play of the single-sided prompt-CLTL game over 
	$(\Int,<,=)$, conforming to strategy $\sstrat'$, is such that 
	$(\sigma', 2\kappa + 2) \models \phi$. This shows that the 
	prompt-CLTL formula $\phi$ is single-sided realizable. 
\end{proof}

The above theorem when combined with Theorem \ref{thm:decresult} 
gives us the following result.
\begin{theorem}
\label{thm:promptCLTLdecresult}
The single-sided realizability problem for prompt-CLTL over $(\Int,<,=)$ is \textsc{2EXPTIME}-complete.
\end{theorem}

\begin{proof}
Theorem~\ref{thm:promptCLTL_to_CLTL_red} implies that the single-sided realizability problem for prompt-CLTL over $(\Int, <, =)$ can be reduced, to the single-sided realizability problem for CLTL. Given a prompt-CLTL formula $\phi$ over a finite set of variables $V$, the CLTL formula $c(\phi)$ over $V \cup \set{\splx, \sply}$ is exponential in the size of $\phi$ as the definition of ${rel}_{\col}(\phi)$ is recursive. However, the number of subformulas of ${rel}_{\col}(\phi)$ is linear in the number of subformulas of $\phi$. Since the single-sided CLTL realizability game is reduced to checking the non-emptiness of a tree automaton (Lemma ~\ref{lem:treeautomata}) whose size depends on the number of subformulas of $c(\phi)$, the single-sided realizability problem for prompt-CLTL is also in \textsc{2EXPTIME}. Now, we know that the realizability problem for LTL is \textsc{2EXPTIME}-complete \cite{VW86} and that every LTL formula is also a prompt-CLTL formula. Therefore, the problem of single-sided realizability for prompt-CLTL is \textsc{2EXPTIME}-complete.
\end{proof}

%% file: conclusion.tex
We have seen in this paper that the CLTL realizability problem 
is decidable over domains satisfying
completion property and that the single-sided CLTL realizability problem is decidable over integers 
with linear order and equality. But both these problems have a high complexity (both are \textsc{2EXPTIME}-complete). It would be interesting to see if there are expressive fragments of CLTL with lower complexity, like the 
fragments of LTL studied in \cite{PPS06}, which work on practical examples.

We believe that single-sided CLTL games over the domain of natural 
numbers $(\Nat,<,=)$ are also decidable. In \cite{DD07}, the authors 
extend the automata-characterization for the satisfiability problem for 
CLTL over the integer domain to the domain of natural numbers. A 
similar extension of the tree-automata characterization for the 
single-sided games over integers to one for single-sided games over 
the naturals seems possible, although the details need to be worked 
out.

Despite the decidability result that we have for the single-sided CLTL 
games over integers, the language of the tree automaton that we 
construct in this paper is an approximation of the set of all winning 
strategy trees. We do not have a machine-theoretic representation for 
winning strategies yet, and this is an interesting direction for future 
exploration.

We have seen that while CLTL games over the integers are undecidable 
in general, restricting to single-sided games yields decidability. It 
would be interesting to see if there are other meaningful restrictions on the 
structure of the game, that yield decidability results.